\documentclass[11pt]{article}

\usepackage{amsfonts,amsmath,amssymb, bbm, bm, xspace, mathrsfs}
\usepackage{multirow,graphicx, epsfig,subfigure,algorithm,algorithmic,rotating}
\usepackage{cite,url}
\usepackage{fullpage}
\usepackage[small,bf]{caption}
\usepackage[table]{xcolor}
\usepackage{tabularx}
\usepackage{psfrag}
\usepackage{subfigure}
\usepackage{verbatim}
\def \endprf{\hfill {\vrule height6pt width6pt depth0pt}\medskip}
\newenvironment{proof}{\noindent {\bf Proof} }{\endprf\par}

\numberwithin{equation}{section}

\newtheorem{theorem}{Theorem}
\newtheorem{proposition}{Proposition}

\newtheorem{lemma}{Lemma}

\newtheorem{definition}{Definition}
\newtheorem{example}{Example}

\newcommand{\mbf}{\boldsymbol}
\newcommand{\mc}{\mathcal}

\newcommand{\mcB}{\mc{B}}

\newcommand{\mcE}{\mc{E}}

\newcommand{\mcI}{\mc{I}}

\newcommand{\mcS}{\mc{S}}

\newcommand{\mcV}{\mc{V}}
\newcommand{\mcZ}{\mc{Z}}

\newcommand{\bfA}{\mbf{A}}
\newcommand{\bfb}{\mbf{b}}

\newcommand{\bfe}{\mbf{e}}
\newcommand{\bfE}{\mbf{E}}

\newcommand{\bfn}{\mbf{n}}

\newcommand{\bfP}{\mbf{P}}
\newcommand{\bfq}{\mbf{q}}
\newcommand{\bfQ}{\mbf{Q}}
\newcommand{\bfr}{\mbf{r}}

\newcommand{\bfs}{\mbf{s}}

\newcommand{\bft}{\mbf{t}}

\newcommand{\bfu}{\mbf{u}}

\newcommand{\bfv}{\mbf{v}}

\newcommand{\bfZ}{\mbf{Z}}
\newcommand{\bfy}{\mbf{y}}
\newcommand{\bfY}{\mbf{Y}}
\newcommand{\bfx}{\mbf{x}}
\newcommand{\bfX}{\mbf{X}}

\newcommand{\bflambda}{\mbf{\lambda}}
\newcommand{\bfgamma}{\mbf{\gamma}}

\newcommand{\bfeta}{\mbf{\eta}}

\newcommand{\bfzero}{\mbf{0}}
\newcommand{\bfone}{\mbf{1}}

\newcommand{\real}{\mathbb{R}}

\newcommand{\integer}{\mathbb{Z}}

\newcommand{\codebook}{\mathcal{C}}

\renewcommand{\SS}{\mathbb{\triangle
}}

\newcommand{\argmin}{\operatorname{argmin}}
\newcommand{\argmax}{\operatorname{argmax}}
\newcommand{\conv}{\operatorname{conv}}
\newcommand{\vect}{\operatorname{vec}}

\newcommand{\st}{\operatorname{subject~to}}

\newcommand{\bfPi}{\mbf{P}}
\newcommand{\bfdelta}{\mbf{\delta}}

\newcommand{\tr}{\mathrm{tr}}

\newcommand{\kendall}{d_K}
\newcommand{\hamming}{d_H}
\newcommand{\chebyshev}{d_{\infty}}

\newcommand{\code}{\mathcal{C}}

\newcommand{\activeset}{\mcS_A}
\newcommand{\clippedset}{\mcS_C}
\newcommand{\zeroset}{\mcS_Z}

\newcommand{\mpmX}{\bfX}
\newcommand{\mpmY}{\bfY}
\newcommand{\mpmXent}{X}
\newcommand{\mpmYent}{Y}

\newcommand{\mpmset}{\mathcal{M}}
\newcommand{\mpx}{\mbf{x}}
\newcommand{\mpxent}{x}
\newcommand{\mpy}{\mbf{y}}
\newcommand{\mpyent}{y}
\newcommand{\multipermpolytope}{\mathbb{M}}
\newcommand{\mpcp}{\mathbb{P}^{\mathsf{M}}}
\newcommand{\mpset}{\mathsf{M}}

\newcommand{\LCMM}{\Pi^{\mathsf{M}}} 
\newcommand{\LCMC}{\Lambda^{\mathsf{M}}} 

\newcommand{\LLRent}{\gamma}
\newcommand{\LLRvect}{\mbf{\gamma}}
\newcommand{\LLRbig}{\mbf{\Gamma}}
\newcommand{\cin}{\mathcal{S}}
\newcommand{\cout}{\Sigma}
\renewcommand{\Pr}{P}
\newcommand{\opmin}{\operatorname{minimize}}
\newcommand{\Lag}{\mathcal{L}}
\newcommand{\nequiv}{\not\equiv}
\newcommand\bigzero{\makebox(0,0){\text{\huge0}}}
\newcommand{\Kc}{\theta}

\DeclareMathOperator*{\eq}{=}
\newcommand{\optbfx}{\mbf{x}}
\newcommand{\optbfz}{\mbf{z}}
\newcommand{\optbfzr}{\mbf{z}^r}
\newcommand{\optbfzc}{\mbf{z}^c}
\newcommand{\simplex}{\SS}
\newcommand{\oneslice}{\mathbb{L}}
\newcommand{\bfPr}{\bfP^r}
\newcommand{\bfPc}{\bfP^c}
\newcommand{\Se}{\mathcal{S}^{e}}
\newcommand{\Sz}{\mathcal{S}^{z}}
\newcommand{\Kz}{\kappa}
\newcommand{\Ke}{\iota}
\newcommand{\size}{A}
\newcommand{\chebyshevball}{V_{\infty}}
\newcommand{\ballsize}{L}
\newcommand{\EX}{\mathbb{E}}

\begin{document}

\title{LP-decodable multipermutation codes\thanks{This material was presented in part at the 2014 Allerton Conference on Communication, Control, and Computing, Monticello, IL, Oct. 2014 and in part at the 2015 Information Theory Workshop (ITW), Jerusalem, Israel, Apr. 2015. This work was supported by the National Science
Foundation (NSF) under Grants CCF-1217058 and by a Natural Sciences and Engineering Research Council of Canada (NSERC) Discovery Research Grant. This paper was submitted to \emph{IEEE Trans. Inf. Theory.}}} 

\author{Xishuo Liu
\thanks{X.~Liu is with the Dept.~of Electrical and Computer
  Engineering, University of Wisconsin, Madison, WI 53706
  (e-mail: xishuo.liu@wisc.edu).}  
, Stark C. Draper
\thanks{S.~C.~Draper is with the Dept.~of Electrical and Computer
  Engineering, University of Toronto, ON M5S 3G4, Canada (e-mail: stark.draper@utoronto.ca).
  }
}

\maketitle

\begin{abstract}
In this paper, we introduce a new way of constructing and decoding multipermutation codes. Multipermutations are permutations of a multiset that generally consist of duplicate entries. We first introduce a class of binary matrices called multipermutation matrices, each of which corresponds to a unique and distinct multipermutation. By enforcing a set of linear constraints on these matrices, we define a new class of codes that we term LP-decodable multipermutation codes. In order to decode these codes using a linear program (LP), thereby enabling soft decoding, we characterize the convex hull of multipermutation matrices. This characterization allows us to relax the coding constraints to a polytope and to derive two LP decoding problems. These two problems are respectively formulated by relaxing the maximum likelihood decoding problem and the minimum Chebyshev distance decoding problem.

Because these codes are non-linear, we also study efficient encoding and decoding algorithms. We first describe an algorithm that maps consecutive integers, one by one, to an ordered list of multipermutations. Based on this algorithm, we develop an encoding algorithm for a code proposed by Shieh and Tsai, a code that falls into our class of LP-decodable multipermutation codes. Regarding decoding algorithms, 
we propose an efficient distributed decoding algorithm based on the alternating direction method of multipliers (ADMM). Finally, we observe from simulation results that the soft decoding techniques we introduce can significantly outperform hard decoding techniques that are based on quantized channel outputs.
\end{abstract}

\section{Introduction}
Using permutations and multipermutations in communication systems dates back at least to~\cite{slepian1965permutation}, where Slepian considered using multipermutations in a data transmission scheme for the additive white Gaussian noise (AWGN) channel. In recent years, there has been growing interest in permutation codes due to their usefulness in various applications such as powerline communications (PLC)~\cite{chu2004constructions} and flash
 memories~\cite{jiang2008rank}. For PLC, permutation codes are proposed to deal with permanent narrow-band noise and impulse noise while delivering constant transmission power (see also~\cite{colbourn2004permutation}). In flash memories, information is stored in the pattern of charge levels of memory cells. Jiang \emph{et al.} proposed using the relative ranking of memory cells to modulate information~\cite{jiang2008rank}. This approach alleviates the over-injection problem during cell programming. In addition, it can reduce errors caused by charge leakage (cf.~\cite{jiang2008rank}). 

Error correction codes that use permutations are usually designed using on a specific distance metric over permutations. In the context of rank modulation, the commonly considered distance metrics include the Kendall tau distance (e.g.,~\cite{jiang2010correcting, barg2010codes, yehezkeally2012snake, mazumdar2013constructions, zhou2014systematic, buzaglo2014perfect}), the Chebyshev distance (e.g.,~\cite{tamo2010correcting, klove2010permutation, schwartz2011optimal, tamo2012on, yehezkeally2012snake, zhou2014systematic}), and the Ulam distance (e.g.,~\cite{farnoud2013error}). Regardless of the choice of distance metric, these studies all consider hard decoding algorithms. In other words, the objective of each decoder is to correct some number of errors in the corresponding distance metric. In order to bring soft decoding to permutation codes, Wadayama and Hagiwara introduce linear programming (LP) decoding of permutation codes in~\cite{wadayama2012lpdecodable}. Although the set of codes that can be decoded by LP decoding is restrictive, the framework is promising for two reasons. First, the algorithm is soft-in soft-out, differentiating itself from hard decoding algorithms based on quantized rankings of channel outputs; and would therefore be expected to achieve lower error rates. Second, the algorithm is based on solving an optimization problem, which makes it possible to incorporate future advances in optimization techniques. 

In this paper, we extend the idea in~\cite{wadayama2012lpdecodable} to multipermutations. Multipermutations generalize permutations by allowing multiple entries of the same value: a multipermutation a permutation of the multiset $\{1,1,\dots,1,2,\dots,2,\dots,m,\dots,m\}$. The number of entries of value $i$ in the multiset is called the \emph{multiplicity} of $i$. We denote by $\bfr = (r_1,\dots,r_m)$ the \emph{multiplicity vector} of a multiset, where $r_i$ is the multiplicity of $i$; in other words, $\bfr$ is the histogram of the multiset. Furthermore, a multipermutation is called \emph{$r$-regular} if $r = r_i$ for all $i$. A multipermutation code can be obtained by selecting a subset of all multipermutations of the multiset. In the literature, multipermutation codes are referred as constant-composition codes
when the Hamming distance is considered~\cite{chu2006on}. When $r_1 = r_2 = \dots = r_m$, the multipermutations under consideration are known as frequency permutation arrays~\cite{huczynska2006frequency}. Recently, multipermutation codes under the Kendall tau distance and the Ulam distance are studied in~\cite{buzaglo2013error} and~\cite{farnoud2014multipermutation} respectively. As mentioned in~\cite{farnoud2014multipermutation}, there are two motivations for using multipermutation codes in rank modulation. First, the size of a codebook based on multipermutations can be larger than that based on permutations (i.e., those in~\cite{zhang2010ldpc}). Second, the number of distinct charges a flash memory can store is limited by the physical resolution of the hardware and thus using permutations over large alphabets is impractical.  

In fact, the construction of Wadayama and Hagiwara in~\cite{wadayama2012lpdecodable} is defined over multipermutations. However in~\cite{wadayama2012lpdecodable}, multipermutations are described using permutation matrices. This results in two notable issues. First, since the size of a permutation matrix scales quadratically with the length of the corresponding vector,
the number of variables needed to specify a multipermutation in this representation scales quadratically with the length of the multipermutation. But since multipermutations consist of many replicated entries, one does not need to describe the relative positions among entries of the same value. This intuition suggests that one can use fewer variables to represent multipermutations. 

The second issue relates to code non-singularity as defined in~\cite{wadayama2012lpdecodable}. To elaborate, we briefly review some concepts therein. In~\cite{wadayama2012lpdecodable}, a codebook is obtained by permuting an initial (row) vector $\bfs$ with a set of permutation matrices. 
If $\bfs$ contains duplicate entries, then there exists at least two different permutation matrices $\bfPi_1$ and $\bfPi_2$ such that $\bfs\bfPi_1 = \bfs\bfPi_2$.  This means that we cannot differentiate between $\bfs\bfPi_1$ and $\bfs\bfPi_2$ by comparing the matrices $\bfPi_1$ and $\bfPi_2$. Due to this ambiguity, permutation matrices are not perfect proxies to multipermutations. To see this, note that the cardinality of a multipermutation code is not necessarily equal to the cardinality of the corresponding permutation matrices. This makes it not straightforward to calculate codebook cardinality from the set of permutation matrices. Furthermore, minimizing the Hamming distance between two multipermutations is not equivalent to minimizing the Hamming distance between two permutation matrices\footnote{This is defined as the number of disagreeing entries (cf. Section~\ref{multipermutation.section.mpm}).}. This can be seen easily from the example above, where the Hamming distance between $\bfs\bfPi_1$ and $\bfs\bfPi_2$ is zero, but the Hamming distance between $\bfPi_1$ and $\bfPi_2$ is greater than zero.

In this paper, we address the above two problems by introducing the concept of multipermutation matrices.
Multipermutation matrices and multipermutations are in a one-to-one relationship. In comparison to a permutation matrix, a multipermutation matrix is a more compact representations of a multipermutation. Further, due to the one-to-one relationship, we can calculate the cardinality of a multipermutation code by calculating the cardinality of the associated multipermutation matrices. In order to construct codes that can be decoded using LP decoding, we develop a simple characterization of the convex hull of multipermutation matrices. The characterization is analogous to the well known Birkhoff polytope (cf.~\cite{marshall2009inequalities}). These results form the basis for the code constructions that follow. They may also be of independent interests to the optimization community. We consider the introduction of multipermutation matrices and the characterization of their convex hull to be our first set of contributions.

Building on these results, our second set of contributions include code definitions and decoding problem formulations. By placing linear constraints on multipermutation matrices to select a subset of multipermutations, we form codebooks that we term \emph{LP-decodable multipermutation codes}. Along this thread, we first present a simple and novel description of a code introduced by Shieh and Tsai in~\cite{shieh2010decoding} (ST codes) that has known rate and distance properties. Then, we study two random coding ensembles and derive their size and distance properties. The code definitions using multipermutation matrices immediately imply an LP decoding formulation. We first relax the maximum likelihood (ML) decoding problem to form an LP decoding problem for arbitrary memoryless channels. In particular, for the AWGN channel, our formulation is equivalent to the decoding problem proposed in~\cite{wadayama2012lpdecodable}. We then relax the minimum (Chebyshev) distance decoding problem to derive an LP decoding scheme that minimizes the Chebyshev distance in a relaxed code polytope. 

Due to the non-linearity of multipermutation codes, we need efficient encoding and decoding algorithms, which brings us to our third set of contributions. To the best of the our knowledge, there has been no encoding algorithms for the ST codes that were introduced in~\cite{shieh2010decoding}. Therefore, we first focus on encoding and introduce an algorithm that ranks all $N = \frac{(\sum_{i = 1}^m r_i)!}{\prod_{i = 1}^m (r_i!)}$ multipermutations that are parameterized by the multiplicity vector $\bfr$. In other words, suppose all multipermutations are ranked such that each corresponds to an index in $\{0,\dots, N-1\}$. Our algorithm outputs the index that corresponds to a given input multipermutation. We use this ranking algorithm as the basis for developing of a low-complexity encoding algorithm for the ST codes. Next, we develop an efficient decoding algorithm based on the alternating direction method of multipliers (ADMM), which has recently been used to develop efficient decoding algorithms for linear codes (e.g,~\cite{barman2013decomposition, liu2014the, liu2014admmfull, yufit2014on, liu2015thesis}). The ADMM decoding algorithm in this paper requires two subroutines that perform Euclidean projections onto two distinct polytopes. Both projections can be solved efficiently, the first using techniques proposed in~\cite{duchi2008efficient} and the second using the algorithm developed in Appendix~\ref{multipermutation.appendix.linear_time_projection}.

We list below our contributions in this paper.
\begin{itemize}
\item We introduce the concept of multipermutation matrices (Section~\ref{multipermutation.subsection.multiperm_matrices}) and characterize the convex hull of all multipermutation matrices (Section~\ref{multipermutation.subsection.convexhull}).
\item We propose LP-decoding multipermutation codes (Section~\ref{multipermutation.subsection.lcmm}) and redefine a code introduced by Shieh and Tsai (ST codes) using our framework (Section~\ref{multipermutation.subsection.examples}). Furthermore, we study two random coding ensembles and compare ST codes with the ensemble average (Section~\ref{multipermutation.subsection.randomcoding} and Appendix~\ref{multipermutation.appendix.random_coding_results}). 
\item We formulate two LP decoding problems (Section~\ref{multipermutation.section.lpdecoding}), one for maximizing the likelihood, and the other for minimizing the Chebyshev distance.
\item We derive an efficient encoding algorithm for ST codes (Section~\ref{multipermutation.subsection.encoding}). 
\item We develop an ADMM decoding algorithm for solving the LP decoding problem for memoryless channels (Section~\ref{multipermutation.subsection.admm}).
\item We initiate the study of initial vector estimation problem for rank modulation and propose a turbo-equalization like decoding algorithm (Appendix~\ref{multipermutation.appendix.initial_vector}).
\end{itemize}
\section{Preliminaries}
\label{multipermutation.section.preliminary}
In this section, we briefly review the concept of permutation matrices and the code construction approach proposed in~\cite{wadayama2012lpdecodable}. 

A length-$n$ permutation $\pi$ is a length-$n$ vector, each element of which is a distinct integer between $1$ and $n$, inclusive. Every permutation corresponds to a unique $n\times n$ permutation matrix, a permutation matrix is a binary matrix such that every row or column sum equals to $1$.
In this paper, all permutations (and multipermutations) are represented using row vectors. Thus, if $\bfPi$ is the permutation matrix corresponding to the permutation $\pi$, then $\pi = 
\imath \bfPi$ where $\imath = (1,2,\dots,n)$ is the identity permutation. We denote by $\Pi_n$ the set of all permutation matrices of size $n\times n$.

\begin{definition}
\label{multipermutation.def.linearly_constrained_permutation_matrix}
(cf.~\cite{wadayama2012lpdecodable})
Let $K$ and $n$ be positive integers. Assume that $\bfA \in \mathbb{Z}^{
K\times n^2}$, $\bfb \in \mathbb{Z}^K$, and let ``$\trianglelefteq
$'' represent a vector of ``$\leq$'' or ``$=$'' relations. A set of \emph{linearly constrained} permutation matrices is defined by 
\begin{equation}
\Pi(\bfA,\bfb,\trianglelefteq
) := \{\bfPi \in \Pi_n | \bfA \vect(\bfPi) \trianglelefteq \bfb \},
\end{equation}
where $\vect(\cdot)$ is the operation of concatenating all columns of a matrix to form a column vector. 
\end{definition}

\begin{definition}
(cf.~\cite{wadayama2012lpdecodable})
\label{multipermutation.def.LP_decodable_perm_code}
Assume the same set up as in Definition~\ref{multipermutation.def.linearly_constrained_permutation_matrix}.
Suppose also that a row vector $\bfs \in \real^n$ is given. The set of vectors $\Lambda(\bfA, \bfb, \trianglelefteq, \bfs)$ given by 
\begin{equation}
\label{multipermutation.eq.lcpc}
\Lambda(\bfA,\bfb,\trianglelefteq
,\bfs) := \{\bfs \bfPi  \in \real^n | \bfPi \in \Pi(\bfA, \bfb, \trianglelefteq)\}
\end{equation}
is called an LP-decodable permutation code\footnote{In this paper, we always let the initial vector be a row vector. Then $\bfs \bfPi$ is again a row vector. This is different from the notation followed by~\cite{wadayama2012lpdecodable}, where the authors consider column vectors.}. $\bfs$ is called the ``initial vector'', which is assumed to be known by both the encoder and the decoder.
\end{definition}

Note that $\bfs$ may contain duplicates, and can be a permutation of the following vector
\begin{equation}
\label{multipermutation.eq.repeating_d}
\bfs = (\underbrace{t_1,t_1,\dots,t_1}_{r_1},\underbrace{t_2,t_2,\dots,t_2}_{r_2},\dots,\underbrace{t_m,t_m,\dots,t_m}_{r_m}),
\end{equation}
where $t_i \neq t_j$ for $i\neq j$ and there are $r_i$ entries with value $t_i$. In this paper, we denote by $\bfr = (r_1,\dots,r_m)$ the multiplicity vector, and let $\bft := (t_1,t_2,\dots,t_m)$. Due to this notation, $\sum_{i = 1}^m r_i = n$. 

Throughout the paper, we use $\bft$ to represent a vector with distinct entries. At this point, it is easy to observe that the vector $\bfs$ can be uniquely determined by $\bft$ and $\bfr$. Therefore, in this paper, we refer $\bft$ as the ``initial vector'' instead of $\bfs$. 

As an important remark, we note that the initial vector does \textbf{not} have to be a vector of integers ($\bfs$ and $\bft$ are in the reals). One can think of initial vectors as the actual charge levels of a flash memory programmed using the rank modulation scheme. For most of this paper, we assume that $\bft$ is fixed and known to the decoder. However, this assumption does \textbf{not} necessarily hold in practice. In Appendix~\ref{multipermutation.appendix.initial_vector}, we briefly discuss some initial work that considers what to do when $\bft$ is unknown to the decoder. 

\section{Multipermutation matrices}
\label{multipermutation.section.mpm}
In this section, we introduce the concept of multipermutation matrices. Although, as in~\eqref{multipermutation.eq.lcpc}, we can obtain a multipermutation of length $n$ by multiplying an length-$n$ initial vector by an $n\times n$ permutation matrix, this mapping is not one-to-one. Thus $|\Lambda(\bfA,\bfb, \trianglelefteq, \bfs)| \leq |\Pi(\bfA,\bfb,\trianglelefteq)|$, where the inequality can be strict if there is at least one $i$ such that $r_i \geq 2$. As a motivating example, let a multipermutation be $\mpx = (1,2,1,2)$. Consider the following two permutation matrices
$$\bfPi_1 = \begin{pmatrix}
1& 0 & 0 & 0\\
0& 0 & 1 & 0\\
0& 1 & 0 & 0\\
0& 0 & 0 & 1
\end{pmatrix}
\text{ and }
\bfPi_2 = \begin{pmatrix}
0& 0 & 1 & 0\\
1& 0 & 0 & 0\\
0& 0 & 0 & 1\\
0& 1 & 0 & 0
\end{pmatrix}.
$$

Then $\mpx = \bfs \bfPi_1 = \bfs \bfPi_2$, where $\bfs = (1,1,2,2)$. In fact there are a total of four permutation matrices that can produce $\mpx$. 

To resolve this ambiguity, we now introduce\emph{ multipermutation matrices}, which are defined to be rectangular binary matrices parameterized by a multiplicity vector. Then, we discuss the advantages of using multipermutation matrices vis-a-vis permutation matrices. Finally, we show a theorem that characterizes the convex hull of multipermutation matrices, a theorem that is crucial to our code constructions and decoding algorithms. 

\subsection{Introducing multipermutation matrices}
\label{multipermutation.subsection.multiperm_matrices}
Recall that $\bfr$ is a multiplicity vector of length $m$ and $n := \sum_{i = 1}^m r_i$. We denote by $\mpset(\bfr)$ the set of all distinct multipermutations parameterized by the multiplicity vector $\bfr$. We now define a set of binary matrices that is in a one-to-one correspondence with $\mpset(\bfr)$.

\begin{definition}
\label{multipermutation.def.multipermutation_matrix}
Given a multiplicity vector $\bfr$ of length $m$ and $n = \sum_{i = 1}^m r_i$, we call a $m \times n$ binary matrix $\mpmX$ a \emph{multipermutation matrix} parameterized by $\bfr$ if $\sum_{i = 1}^m \mpmXent_{ij} = 1$ for all $j$ and $\sum_{j = 1}^n \mpmXent_{ij} = r_i$ for all $i$. Denote by $\mpmset(\bfr)$ the set of all multipermutation matrices parameterized by $\bfr$. 
\end{definition}

Using this definition, it is easy to build a bijective mapping between multipermutations and multipermutation matrices. When the initial vector is $\bft$, the mapping $\mpset(\bfr)\mapsto\mpmset(\bfr)$ can be defined as follows: Let $\mpx$ denote a multipermutation. Then, it is uniquely represented by the multipermutation matrix $\mpmX$ such that $\mpmXent_{ij} = 1$ if and only if $\mpxent_j = t_i$. Conversely, to obtain the multipermutation $\mpx$, one can simply calculate the product $\bft\mpmX$.

\begin{example}
Let the multiplicity vector be $\bfr = (2,3,2,3)$, let $\bft = (1,2,3,4)$, and let $\mpx = (2, 1, 4, 1, 2, 3, 4, 4, 2, 3).$ Then the corresponding multipermutation matrix is $$\mpmX = 
\begin{pmatrix}
0 & 1 & 0 & 1 & 0 & 0 & 0 & 0 & 0 & 0\\
1 & 0 & 0 & 0 & 1 & 0 & 0 & 0 & 1 & 0\\
0 & 0 & 0 & 0 & 0 & 1 & 0 & 0 & 0 & 1\\
0 & 0 & 1 & 0 & 0 & 0 & 1 & 1 & 0 & 0
\end{pmatrix}.
$$
The row sums of $\mpmX$ are $(2,3,2,3)$ respectively. Further, $\mpx = \bft \mpmX$.
\end{example}

\begin{lemma}
\label{multipermutation.lemma.multiperm_one_to_one}
Let $\bft$ be an initial vector of length $m$ with $m$ distinct entries. Let $\mpmX$ and $\mpmY$ be two multipermutation matrices parameterized by a multiplicity vector $\bfr$. Further, let $\mpx = \bft \mpmX$ and $\mpy = \bft \mpmY$. Then $\mpx = \mpy$ if and only if $\mpmX = \mpmY$.
\end{lemma}
\begin{proof}
First, it is obvious that if $\mpmX = \mpmY$ then $\mpx = \mpy$. Next, we show that if $\mpmX \neq \mpmY$ then $\mpx \neq \mpy$. We prove by contradiction. Assume that there exists two multipermutation matrices $\mpmX$ and $\mpmY$ such that $\mpmX \neq \mpmY$ and $\mpx = \mpy$. Then
\begin{align*}
\mpx - \mpy = \bft (\mpmX - \mpmY).
\end{align*}
Since $\mpmX \neq \mpmY$, there exists at least one column $j$ such that $\mpmX$ has a $1$ at the $k$-th row and $\mpmY$ has a $1$ at the $l$-th row where $k \neq l$. Then the $j$-th entry of $\bft (\mpmX - \mpmY)$ would be $t_k - t_l \neq 0$ because all entries of $\bft$ are different. This contradicts the assumption that $\mpx - \mpy = \bfzero$.
\end{proof}

At this point, one may wonder why this one-to-one relationship matters. We now discuss three aspects in which having a one-to-one relationship between multipermutations and multipermutation matrices is beneficial. 

\subsubsection{Reduction in the number of variables}
One immediate advantage of using multipermutation matrices is that they require fewer variables to represent multipermutations. (This was the first issue discussed in the Introduction.) The multipermutation matrix corresponding to a length-$n$ multipermutation has size $m\times n$, where $m$ is the number of distinct values in the multipermutation. 

This benefit can be significant when the multiplicities are large, i.e., when $m$ is much smaller than $n$. For example, a triple level cell flash memory has $8$ states per cell. If a multipermutation code is of length $100$, then one needs an $8\times 100$ multipermutation matrix to represent a codeword. The corresponding permutation matrix has size $100\times 100$.

\subsubsection{Analyzing codes via binary matrices}
Also due to the one-to-one relationship, one can use multipermutation matrices as a proxy for multipermutations. By this we mean that one can analyze properties of multipermutation codes by analyzing the associating multipermutation matrices. First, note that a set of multipermutation matrices of cardinality $M$ can be mapped to a set of multipermutations of cardinality $M$. This means that one can determine the size of a multipermutation code by characterizing the cardinality of the corresponding set of multipermutation matrices. As an example, in Section~\ref{multipermutation.subsection.randomcoding}, we analyze the average cardinality of two random coding ensembles. Second, one can determine distance properties of a multipermutation code via the set of multipermutation matrices. One such example is the Hamming distance, which is demonstrated in details in the following. 

\subsubsection{The Hamming distance}
The Hamming distance between two multipermutations is defined as the number of entries in which the two vectors differ from each other. More formally, let $\mpx$ and $\mpy$ be two multipermutations, then $\hamming(\mpx, \mpy) = |\{i| \mpxent_i \neq \mpyent_i\}|$. Due to Lemma~\ref{multipermutation.lemma.multiperm_one_to_one}, we can express the Hamming distance between two multipermutations using their corresponding multipermutation matrices. 
\begin{lemma}
\label{multipermutation.lemma.Hamming}
Let $\mpmX$ and $\mpmY$ be two multipermutation matrices, and let $\bft$ be an initial vector with distinct entries. With a small abuse of notations, denote by $\hamming(\mpmX, \mpmY)$ the Hamming distance between the two matrices, which is defined by $\hamming(\mpmX, \mpmY) := |\{(i,j)| \mpmXent_{ij} \neq \mpmYent_{ij}\}|$. Then 
$$\hamming(\mpmX, \mpmY) = 2 \hamming(\mpx, \mpy),$$
where $\mpx = \bft \mpmX$ and $\mpy = \bft \mpmY$. Furthermore,
$$\hamming(\mpmX, \mpmY) = \tr(\mpmX^T (\bfE - \mpmY)),$$
where $\tr(\cdot)$ represents the trace of the matrix and $\bfE$ is an $n\times n$ matrix with all entries equal to $1$. \footnote{We note that $\tr(\bfX^T \bfY) = \sum_{i,j} X_{ij}Y_{ij}$ is the Frobenius inner product of two equal-sized matrices.}
\end{lemma}
\begin{proof}
For all $j$ such that $\mpxent_j \neq \mpyent_j$, the $j$-th column of $\mpmX$ differs from the $j$-th column of $\mpmY$ by two entries. As a result, the distance between multipermutation matrices is double the distance between the corresponding multipermutations.

Next, 
$\tr(\mpmX^T (\bfE - \mpmY)) = \sum_{ij} \mpmXent_{ij} (1 - \mpmYent_{ij})$. 
If $\mpmXent_{ij} = \mpmYent_{ij}$ then $\mpmXent_{ij} (1 - \mpmYent_{ij}) = 0$. Otherwise $\mpmXent_{ij} (1 - \mpmYent_{ij}) = 1$. Therefore $\hamming(\mpmX, \mpmY) = \tr(\mpmX^T (\bfE - \mpmY))$.
\end{proof}

The above two points relate to the second issue regarding the redundancy in the representation of~\cite{wadayama2012lpdecodable}, as was discussed in the Introduction.

\subsection{Geometry of multipermutation matrices}
\label{multipermutation.subsection.convexhull}
In this section we prove an important theorem that characterizes the convex hull of all multipermutation matrices. As background, we first review the definition of doubly stochastic matrices. Then, we state the Birkhoff-von Neumann theorem for permutation matrices. Finally, we build off the Birkhoff-von Neumann theorem to prove our theorem. We refer readers to~\cite{marshall2009inequalities} and references therein for more materials on doubly stochastic matrices and related topics.
\begin{definition}
An $n\times n$ matrix $\bfQ$ is doubly stochastic if 
\begin{itemize}
\item[$(a)$] $Q_{ij} \geq 0$; 
\item[$(b)$] $\sum_{i = 1}^n Q_{ij} = 1$ for all $j$ and $\sum_{j = 1}^n Q_{ij} = 1$ for all $i$.
\end{itemize}
\end{definition}

The set of all doubly stochastic matrices is called the Birkhoff polytope. There is a close relationship between the Birkhoff polytope and the set of permutation matrices as the following theorem formalizes:
\begin{theorem}(Birkhoff-von Neumann Theorem, cf.~\cite{marshall2009inequalities})
\label{multipermutation.theorem.birkhoff}
The permutation matrices constitute the extreme points of the set of doubly stochastic matrices. Moreover, the set of doubly stochastic matrices is the convex hull of the permutation matrices.
\end{theorem}

This theorem is the basis for the decoding problem formulated in~\cite{wadayama2012lpdecodable}. Namely, the LP relaxation for codes defined by Definition~\ref{multipermutation.def.LP_decodable_perm_code} is based on the Birkhoff polytope. In order to formulate LP decoding problems using multipermutation matrices, we need a similar theorem that characterizes the convex hull of multipermutation matrices. 

Denote by $\multipermpolytope(\bfr)$ the convex hull of all multipermutation matrices parameterized by $\bfr$, i.e. $\multipermpolytope(\bfr) = \conv(\mpmset(\bfr))$. Then, $\multipermpolytope(\bfr)$ is characterized by the following theorem.

\begin{theorem}
\label{multipermutation.theorem.multiperm_polytope}
Let $\bfr\in \integer_+^m$ and $\bfZ$ be an $m\times n$ matrix such that 
\begin{itemize}
\item[$(a)$] $\sum_{i = 1}^m Z_{ij} = 1$ for all $j = 1,\dots,n$.
\item[$(b)$] $\sum_{j = 1}^n Z_{ij} = r_i$ for all $i = 1,\dots,m$.
\item[$(c)$] $Z_{ij} \in [0,1]$ for all $i$ and $j$.
\end{itemize}
Then, $\bfZ$ is a convex combination of all multipermutation matrices parameterized by $\bfr$. Conversely, any convex combination of multipermutation matrices parameterized by $\bfr$ satisfies the above conditions.
\end{theorem}  
\begin{proof}
Consider a multipermutation $\mpx = (\mpxent_1, \dots, \mpxent_n)$ where each $\mpxent_k \in \{1,\dots,m\}$. Without loss of generality, we assume that $\mpx$ is in increasing order. We denote by $\mcI_i$ the index set for the $i$-th symbol, i.e.,
\begin{equation}
\label{multipermutation.eq.indexset}
\mcI_i := \left\{\sum_{l = 1}^{i-1}r_l +1 ,\dots, \sum_{l = 1}^{i}r_l\right\}.
\end{equation}
Then $\mpxent_k = i$ if $ k \in \mcI_i$. Let $\mpmX$ be the corresponding $m\times n$ multipermutation matrix. Then $\mpmX$ has the following form
\begin{equation*}
\mpmX = \begin{pmatrix}
\underbrace{1 \; \dots\; 1}_{r_1} &  &&& \bigzero \\
& \underbrace{1 \; \dots\; 1}_{r_2} &&&\\
& & & \ddots &\\
\bigzero& & & & \underbrace{1 \; \dots\; 1}_{r_m}
\end{pmatrix},
\end{equation*}
where $X_{ik} = 1$ if $k \in \mcI_i$ and  $X_{ik} = 0$ otherwise.

Note that all multipermutation matrices parameterized by a fixed $\bfr$ are column permutations of each other. Of course, as already pointed out, not all distinct permutations of columns yield distinct multipermutation matrices. To show that any $\bfZ$ satisfying $(a)$-$(c)$ is a convex combination of multipermutation matrices, we show that there exists an $n\times n$ stochastic matrix $\bfQ$ such that $\bfZ = \mpmX \bfQ$. Then by Theorem~\ref{multipermutation.theorem.birkhoff}, $\bfQ$ can be expressed as a convex combination of permutation matrices. In other words, $\bfQ = \sum_{h}\alpha_h \bfPi_h$ where $\bfPi_h \in \Pi_n$ are permutation matrices; $\alpha_h \geq 0$ for all $h$ and $\sum_{h} \alpha_h = 1$. Then we have 
\begin{align*}
\bfZ = \mpmX  \sum_{h}\alpha_h \bfPi_h =  \sum_{h}\alpha_h (\mpmX \bfPi_h),
\end{align*}
where $\mpmX \bfPi_h$ is a column permuted version of the matrix $\mpmX$, which is a multipermutation matrix of multiplicity $\bfr$. This implies that $\bfZ$ is a convex combination of multipermutation matrices.

We construct the required $n \times n$ matrix $\bfQ$ in the following way. For each $i \in (1,2,\dots,m)$, let $\bfq^i$ be a length-$n$ \emph{column} vector, $q^i_j = \frac{1}{r_i} Z_{ij}$ for $j = 1,\dots,n$.  Then the $n\times n$ matrix
\begin{equation}
\bfQ^T := \big[\underbrace{\bfq^{1} | \bfq^{1} |\dots|}_{r_{1} \text{ of them}} \dots \underbrace{| \bfq^i | \bfq^i | \dots|}_{r_i \text{ of them}}
\dots \underbrace{|\bfq^{m} |\bfq^{m} \dots }_{r_{m} \text{ of them}} \big ].
\end{equation}
In other words, $Q_{kj}  = \frac{1}{r_i} Z_{ij}$ for all $k \in \mcI_{i}$ and $j = 1,\dots,n$.
We now verify that $\bfZ = \mpmX \bfQ$  and that $\bfQ$ is doubly stochastic, which by our discussions above implies that $\bfZ$ is a convex combination of column-wise permutations of $\mpmX$.
\begin{enumerate}
\item To verify $\bfZ = \mpmX \bfQ$, we need to show that $Z_{ij} = \sum_{k = 1}^n\mpmXent_{ik}Q_{kj}$. Since $\mpmX$ is a binary matrix,
\begin{align*}
\sum_{k= 1}^n \mpmXent_{ik}Q_{kj} = \sum_{k:\mpmXent_{ik} = 1} Q_{kj}.
\end{align*}
In addition, since $\mpx$ is sorted, $\mpmXent_{ik} = 1$ if and only if $k \in \mcI_{i}$.
By the definition of $\bfQ$, $Q_{kj}  = \frac{1}{r_i} Z_{ij}$ for all $k \in \mcI_{i}$. Therefore 
\begin{align*}
\sum_{k:\mpmXent_{ik} = 1} Q_{kj} = r_{i} \frac{Z_{ij}}{r_{i}} =Z_{ij}.
\end{align*}

\item Next we verify that $\bfQ$ is a double stochastic matrix. Since $0 \leq Z_{ij} \leq 1$ for all $i,j$,  $ Q_{ij} \geq 0$ for all $i,j$. By the definition of $\bfQ$, the sum of each row is $\|\bfq^i\|_1$ for some $i$. Thus $\|\bfq^i\|_1 = \sum_{j = 1}^n \frac{1}{r_i}Z_{ij} = 1$ by condition $(b)$. The sum of each column is 
\begin{align*}
\sum_{k = 1}^n Q_{kj} &=\! \sum_{i = 1}^m \sum_{k\in\mcI_i}  Q_{kj}\\
 &= \!\sum_{i = 1}^m \sum_{k\in\mcI_i} \frac{1}{r_i} Z_{ij} = \!\sum_{i = 1}^m Z_{ij} = 1,
\end{align*}
where the last equality is due to condition $(a)$.
\end{enumerate}
To summarize, for any given real matrix $\bfZ$ satisfying condition $(a)$-$(c)$ we can find a doubly stochastic matrix $\bfQ$ such that $\bfZ = \mpmX \bfQ$ for a particular multipermutation matrix $\mpmX$. This implies that $\bfZ$ is a convex combination of multipermutation matrices.

The converse is easy to verify by the definition of convex combinations and therefore is omitted.
\end{proof}

\section{LP-decodable multipermutation code}
\label{multipermutation.section.lpd_MP_code}
\subsection{Constructing codes using linearly constrained multipermutation matrices}
\label{multipermutation.subsection.lcmm}
Using multipermutation matrices as defined in Definition~\ref{multipermutation.def.multipermutation_matrix}, we define the set of linearly constrained multipermutation matrices analogous to that in~\cite{wadayama2012lpdecodable}\footnote{The analogy is in the following sense. One can obtain the definitions in this section by restating the definition in~\cite{wadayama2012lpdecodable} using multipermutations and the convex hull $\multipermpolytope(\bfr)$.s}. 
\begin{definition}
\label{multipermutation.def.linearly_constrained_multipermutation_matrix}
Let $\bfr$ be a length-$m$ multiplicity vector, and $n:= \sum_{i = 1}^m r_i$. Let $K$ be a positive integer. Assume that $\bfA \in \mathbb{Z}^{K \times (mn)}$, $\bfb \in \mathbb{Z}^K$, and let ``$\trianglelefteq
$'' represent a vector of ``$\leq$'' or ``$=$'' relations. A set of linearly constrained multipermutation matrices is defined as 
\begin{equation}
\LCMM(\bfr,\bfA,\bfb,\trianglelefteq
) := \{\bfX \in \mpmset(\bfr) | \bfA \vect(\bfX) \trianglelefteq \bfb \},
\end{equation}
where $\mpmset(\bfr)$ is the set of all multipermutation matrices parameterized by $\bfr$.
\end{definition}

\begin{definition}
\label{multipermutation.def.LP_decodable_multiperm_code}
Let $\bfr$ be a length-$m$ multiplicity vector, and $n:= \sum_{i = 1}^m r_i$. Let $K$ be a positive integer. Assume that $\bfA \in \mathbb{Z}^{K \times (mn)}$, $\bfb \in \mathbb{Z}^K$, and let ``$\trianglelefteq
$'' represent a vector of ``$\leq$'' or ``$=$'' relations. 
Suppose also that $\bft \in \real^m$ is given. The set of vectors $\LCMC(\bfr, \bfA, \bfb, \trianglelefteq, \bft)$ given by 
\begin{equation}
\label{multipermutation.eq.lcmc}
\LCMC(\bfr, \bfA,\bfb,\trianglelefteq
,\bft) := \{\bft \bfX \in \real^n | \bfX \in \LCMM(\bfr, \bfA, \bfb, \trianglelefteq)\}
\end{equation}
is called an LP-decodable multipermutation code.
\end{definition}

We can relax the integer constraints and form a \emph{code polytope}. Recall that $\multipermpolytope(\bfr)$ is the convex hull of all multipermutation matrices parameterized by $\bfr$.
\begin{definition}
\label{multipermutation.def.code_polytope}
The polytope $\mpcp(\bfr, \bfA,\bfb,\trianglelefteq)$ defined by
\begin{equation*}
\mpcp(\bfr, \bfA,\bfb,\trianglelefteq) := \multipermpolytope(\bfr) \bigcap \{\bfX \in \real^{m\times n} | \bfA \vect(\bfX) \trianglelefteq \bfb\}
\end{equation*}
is called the ``code polytope''. We note that $\mpcp(\bfr, \bfA,\bfb,\trianglelefteq)$ is a polytope because it is the intersection of two polytopes.
\end{definition}

Regarding the above definitions, we discuss some key ingredients.
\begin{itemize}
\item Definition~\ref{multipermutation.def.linearly_constrained_multipermutation_matrix} defines the set of multipermutation matrices. Due to Lemma~\ref{multipermutation.lemma.multiperm_one_to_one}, this set uniquely determines a set of multipermutations. The actual codeword that is transmitted (or stored in a memory system) is also determined by the initial vector $\bft$, which depends on the modulation scheme used in the system. Definition~\ref{multipermutation.def.LP_decodable_multiperm_code} is the set of codewords determined by $\LCMM(\bfr,\bfA,\bfb,\trianglelefteq)$ once the initial vector $\bft$ is into account.
\item  Definition~\ref{multipermutation.def.code_polytope} is useful for decoding. It will be discussed in detail in Section~\ref{multipermutation.section.lpdecoding}. As a preview, in Section~\ref{multipermutation.section.lpdecoding}, we will formulate two optimization problems with variables constrained by the code polytope $\mpcp(\bfr, \bfA,\bfb,\trianglelefteq)$. In both optimizations, the objective functions will be related to the initial vector $\bft$ but the constraints will only be a function of $\mpcp(\bfr, \bfA,\bfb,\trianglelefteq)$. We emphasize that $\mpcp(\bfr, \bfA,\bfb,\trianglelefteq)$ is not parameterized by $\bft$.
\item $\mpcp(\bfr, \bfA,\bfb,\trianglelefteq)$ is defined as the intersection of two polytopes. It is not defined as the convex hull of $\LCMC(\bfr, \bfA,\bfb,\trianglelefteq,\bft)$, which is usually hard to describe. However, the intersection that define $\mpcp(\bfr, \bfA,\bfb,\trianglelefteq)$ may introduce fractional vertices, i.e., vertices $\bfX \in \real^{m\times n}$ such that $X_{ij} \in (0,1)$. Because of this, we call $\mpcp(\bfr, \bfA,\bfb,\trianglelefteq)$ a relaxation of $\conv(\LCMC(\bfr, \bfA,\bfb,\trianglelefteq,\bft))$. 
\end{itemize}

To better explore structures of LP-decodable multi-permutation codes, we now define two specific types of linear constraints.
\begin{definition}
\label{multipermutation.def.fixed-at-something}
\textbf{Fixed-at-zero} constraints: Let $\mcZ$ be a set of entries $(i,j)$. A code with a set of fixed-at-zero constraints is defined by both $\mpmX \in \mpmset(\bfr)$ and $\mpmXent_{ij} = 0$ for all $(i,j) \in \mcZ$.
\textbf{Fixed-at-equality} constraints: Let $\mcE$ be a set of entry pairs $(i,j),(k,l)$. A code with a set of fixed-at-equality constraint is defined by both $\mpmX \in \mpmset(\bfr)$ and $\mpmXent_{ij} = \mpmXent_{kl}$ for all $(i,j),(k,l) \in \mcE$. These two types of constraints can be combined.
\end{definition}

We consider these two types of constraints because they are useful to define LP-decodable multipermutation and permutation codes, and because we can develop efficient decoding algorithms for these codes. For example, the pure ``involution'' code introduced in~\cite{wadayama2012lpdecodable} is constructed by combining both constraints. In Section~\ref{multipermutation.subsection.examples}, we discuss two codes constructed using fixed-at-zero constraints. In Section~\ref{multipermutation.subsection.randomcoding}, we show random coding results for these two types of codes. Last but not least, we show how to decode codes with fixed-at-zero, fixed-at-equality, or both constraints using ADMM in Section~\ref{multipermutation.subsection.admm}.

\subsubsection*{Remarks}
A natural question to ask is whether the restriction to linear constraints reduces the space of possible code designs. In our previous paper~\cite{liu2014lp}, we show that the answer to this question is ``No.'' This follows because it is possible to define an arbitrary codebook using linear constraints. As we show more formally in~\cite{liu2014lp}, one can add one linear constraint for each non-codeword, where the linear constraint requires that the Hamming distance between any codeword and that non-codeword to be at least $1$. However, this approach leads to an exponential growth in the number of linear constraints. Thus, the interesting and challenging question is how to construct \textbf{good} codes (in terms of rate and error performance) that can be described \textbf{efficiently} using linear constraints. A related question that we study in~\cite{liu2014lp} connects the description of LP-decodable permutation code (as defined in Definition~\ref{multipermutation.def.LP_decodable_perm_code}) to that of LP-decodable multipermutation code (as defined in Definition~\ref{multipermutation.def.LP_decodable_multiperm_code}). We show that codes described by Definition~\ref{multipermutation.def.LP_decodable_multiperm_code} can be restated using Definition~\ref{multipermutation.def.LP_decodable_perm_code} using the same number of linear constraints (the same ``$K$'' in Definition~\ref{multipermutation.def.LP_decodable_perm_code} and~\ref{multipermutation.def.LP_decodable_multiperm_code}). We refer readers to~\cite{liu2014lp} for the details of these results.

\subsection{Examples of LP-decodable multipermutation codes}
\label{multipermutation.subsection.examples}
We provide two examples of codes using Definition~\ref{multipermutation.def.LP_decodable_multiperm_code}. 
\begin{example}[Derangement] A permutation $\pi$ is termed a ``derangement'' if $\pi_i \neq i$ for all $i\in \{1,\dots,n\}$. For multipermutations, we define a generalized notion of derangement as follows. Let $$\imath = (\underbrace{1, 1,\dots, 1}_{r_1},\underbrace{ 2, 2,\dots, 2}_{r_2},\dots,\underbrace{ m, m,\dots, m}_{r_m}).$$
Let $\mpx$ be a multipermutation obtained by permuting $\imath$. We say that $\mpx$ is a derangement if $\mpxent_i \neq \imath_i$ for all $i$.

In~\cite{wadayama2012lpdecodable}, the authors use Definition~\ref{multipermutation.def.LP_decodable_perm_code} to define the set of derangements by letting $\tr(\bfP) = 0$, where $\bfP$ is a permutation matrix. We now extend this construction using Definition~\ref{multipermutation.def.LP_decodable_multiperm_code} and let the linear constraints on the multipermutation matrix $\bfX$ be
\begin{equation}
\label{multipermutation.eq.derangement}
\sum_{j \in \mcI_i} X_{ij} = 0\text{ for all }i = 1,\dots,m,
\end{equation}
where $\mcI_i$ is defined by~\eqref{multipermutation.eq.indexset}. Suppose the initial vector $\bft = (1,2,\dots,m)$, then~\eqref{multipermutation.eq.derangement} implies that symbol $i$ cannot appear at positions $\mcI_i$. For example, let $\bft = (1,2,3)$ and $\bfr = (2,2,2)$. Then the allowed derangements that form the codebook are 
\begin{align*}
&(3, 3, 1, 1, 2, 2), (2, 2, 3, 3, 1, 1), (2, 3, 1, 3, 2, 1), \\
&(2, 3, 1, 3, 1, 2), (2, 3, 3, 1, 2, 1), (2, 3, 3, 1, 1, 2),\\
&(3, 2, 1, 3, 2, 1), (3, 2, 1, 3, 1, 2), (3, 2, 3, 1, 2, 1), \\
&(3, 2, 3, 1, 1, 2).
\end{align*}
\end{example}

\begin{example}
\label{multipermutation.example.chebyshev_code}
In~\cite{shieh2010decoding}, Shieh and Tsai study multipermutation codes under the Chebyshev distance. The Chebyshev distance between two multipermutations $\mpx$ and $\mpy$ is defined as 
\begin{equation}
\chebyshev (\mpx, \mpy) = \max_{i} |\mpxent_i - \mpyent_i|.
\end{equation}
We review the Shieh-Tsai (ST) code (cf.~\cite[Construction~1]{shieh2010decoding}) in Definition~\ref{multipermutation.def.chebyshev_code}.
\begin{definition}
\label{multipermutation.def.chebyshev_code}
Let $\bfr = (r,r,\dots,r)$ be a length-$m$ vector. Let $d$ be an integer such that $d$ divides $m$. We define
\begin{equation}
\codebook(r,m,d) = \{\bfx \in \mpset(\bfr)| \forall i \in \{1,\dots,mr\}, x_i \equiv i \bmod{d} \}.
\end{equation}
\end{definition}
Although not originally presented that way in~\cite{shieh2010decoding}, it is easy to verify that this code is an LP-decodable multipermutation code defined by fixed-at-zero constraints. The fixed-at-zero constraints are defined by the set $\mcZ= \{(i,j) | j = 1,\dots,n \text{ and } i \nequiv j \bmod{d}\}. $ As a concrete example, let $m = 6$, $r = 2$ and $d = 3$. Then the constraints are
\begin{align*}
X_{21} = X_{31} = X_{51} = X_{61} &= 0\\
X_{12} = X_{32} = X_{42} = X_{62} &= 0\\
\vdots& \\
X_{1, 12} = X_{2, 12} = X_{4, 12} = X_{5, 12} &= 0.
\end{align*}

It is showed in~\cite{shieh2010decoding} that this code has cardinality $(\frac{(ar)!}{(r!)^a})^d$ where $a = m/d$. Further, the minimum Chebyshev distance of this code is $d$. In addition, for large values of $r$, the rate of the code is observed to be close to a theoretical upper bound on all codes of Chebyshev distance $d$. However, no encoding or decoding algorithms are presented in~\cite{shieh2010decoding}. We discuss encoding and decoding algorithms for this code in Section~\ref{multipermutation.section.algorithms}.
\end{example}

\subsection{The random coding ensemble}
\label{multipermutation.subsection.randomcoding}
In this subsection, we study randomly constructed LP-decodable multipermutation codes. We focus on the ensembles generated either by fixed-at-zero constraints or by fixed-at-equality constraints. The randomness comes from choosing the respective constraint sets, $\mcZ$ or $\mcE$, uniformly at random. Unfortunately, as we show in Appendix~\ref{multipermutation.appendix.numerical_randomcoding}, several results indicate that the ensemble average is not as good as ST codes, which are structured codes belonging to the ensemble. Therefore, we only briefly present our problem formulations and results in the main text and refer readers to Appendix~\ref{multipermutation.appendix.random_coding_results} for more details.

We first introduce some additional notation. With a small abuse of notation, we denote by $\LCMM(\bfr,\mcZ)$ the set of multipermutation matrices constrained by fixed-at-zero constraints. Similarly, denote by $\LCMM(\bfr,\mcE)$ the set of multipermutation matrices constrained by fixed-at-equality constraints.
Denote by $\Kz$ and $\Ke$ the respective cardinalities of sets $\mcZ$ and $\mcE$. Furthermore, denote by $\Sz(\Kz)$ the set of all possible choices of $\mcZ$ that have $\Kz$ elements; denote by $\Se(\Ke)$ the set of all possible choices of $\mcE$ that have $\Ke$ elements. 

Note that we do not consider duplicated constraints. In other words, all entries in $\mcE$ (or $\mcZ$) are distinct from each other. This is different from the set up in~\cite[Sec. VI]{wadayama2012lpdecodable}, where the authors allow repeated constraints. Consequently, the cardinalities of both types of constraints, i.e., $\Kz$ and $\Ke$, are limited. For example, since there are $mn - n$ zeros in a multipermutation matrix, $\Kz$ should be less than or equal to $mn - n$; otherwise the cardinality of the code must be zero\footnote{The cardinality of a code may be zero even when $\Kz$ is small, e.g., when $\mcZ$ fixes a whole column to zero. But for $\Kz\geq mn - n$ distinct constraints, the code size is zero regardless how we pick $\mcZ$.}. 
On the other hand, fixed-at-equality constraints are constructed by entry pairs. There are $\binom{nm}{2}$ ways of choosing two entries from a multipermutation matrix. Therefore $\Ke \leq \binom{nm}{2}$.
\begin{lemma}
\label{multipermutation.lemma.all_possible_choices}
$$|\Sz(\Kz)| = \binom{nm}{\Kz}, \quad |\Se(\Ke)| = \binom{\binom{nm}{2}}{\Ke}.$$
\end{lemma}

Next, we draw $\mcZ$ (resp. $\mcE$) uniformly at random from the set $\Sz(\Kz)$ (resp. $\Se(\Ke)$). As a result, for particular realizations $\mcZ$ and $\mcE$, $\Pr(\mcZ) = \frac{1}{|\Sz(\Kz)|}$ and $\Pr(\mcE) = \frac{1}{|\Se(\Ke)|}$. When taking into account all possible choices, we can show the following lemma for multipermutation matrices.
\begin{lemma}
\label{multipermutation.lemma.codebook_symmetry}
Consider a \emph{fixed} multiplicity vector $\bfr$ and a \emph{fixed} multipermutation matrix $\mpmX \in \mpmset(\bfr)$. Let $\Kz$ and $\Ke$ be fixed parameters, then
$$
|\{ \mcZ \in \Sz(\Kz)  |  \mpmX \in \LCMM(\bfr,\mcZ)    \}| = \binom{nm -n}{\Kz}.
$$ 
Similarly,
$$
|\{ \mcE \in \Se(\Ke)  |  \mpmX \in \LCMM(\bfr,\mcE)    \}| = \binom{\binom{nm -n}{2} + \binom{n}{2}}{\Ke}.
$$ 
\end{lemma}

\begin{proof}
See Appendix~\ref{multipermutation.appendix.proof.codebook_symmetry}.
\end{proof}

Following the methodology adopted in~\cite{wadayama2012lpdecodable}, we prove Proposition~\ref{multipermutation.proposition.codesize} which calculates the average cardinality of multipermutation matrices that meets a randomly chosen set of either fixed-at-zero or fixed-at-equality constraints. Note that due to Lemma~\ref{multipermutation.lemma.multiperm_one_to_one}, Proposition~\ref{multipermutation.proposition.codesize} actually calculates the codebook size. 
\begin{proposition}
\label{multipermutation.proposition.codesize}
Denote by $\size(\LCMM(\bfr,\mcZ))$ the cardinality of the code. Then,
$$
\EX[\size(\LCMM(\bfr,\mcZ))] = \frac{\binom{nm -n}{\Kz}|\mpmset(\bfr)|}{|\Sz(\Kz)|},
$$
where the expectation is taken over all possible choices of $\mcZ \in \Sz(\Kz)$. Using the same notation, 
$$
\EX[\size(\LCMM(\bfr,\mcE))] = \frac{\binom{\binom{nm -n}{2} + \binom{n}{2}}{\Ke} |\mpmset(\bfr)|}{|\Se(\Ke)|},
$$
where the expectation is taken over all possible choices of $\mcE \in \Se(\Ke)$. Recall that $|\mpmset(\bfr)| = \frac{n!}{\prod_{i = 1}^m(r_i!)}$. Further, $|\Se(\Ke)|$ and $|\Sz(\Kz)|$ can be calculated using Lemma~\ref{multipermutation.lemma.all_possible_choices}.
\end{proposition}
\begin{proof}
See Appendix~\ref{multipermutation.appendix.proof.codesize}.
\end{proof}

We now study the distance properties of these codes. We are particularly interested in the Chebyshev distance of $r$-regular multipermutations, for we can directly compare our results to the distance property of ST codes. Let $\mpy$ be a fixed multipermutation that may or may not be a codeword. Following the terminology used in~\cite{wadayama2012lpdecodable}, we refer $\mpy$ as the ``origin'' multipermutation; we consider the Chebyshev distance from the fixed $\mpy$ to other codewords. We use $\bft = (1,2,\dots,m)$ as the initial vector. 
\begin{proposition}
\label{multipermutation.proposition.average_ball}
Let $\bfr = (r,r,\dots,r)$ and define
$$\ballsize_d(\LCMM(\bfr,\mcZ)) := |\{\mpmX \in \LCMM(\bfr,\mcZ)| \chebyshev(\bft\mpmX, \mpy) \leq d \}|.$$
Then,
\begin{equation}
\label{multipermutation.eq.ballsize_fixedatzero}
\begin{split}
 \frac{\binom{mn-n}{\Kz}}{|\Sz(\Kz)|}  \frac{(2dr+r)^n n!}{2^{2dr}n^n (r!)^m} &\leq \EX[\ballsize_d(\LCMM(\bfr,\mcZ))] \\ &\leq \frac{\binom{mn-n}{\Kz}}{|\Sz(\Kz)|} \frac{[(2dr+r)!]^{\frac{n}{2dr+r}}}{(r!)^m}
\end{split}
\end{equation}
where the expectation is taken over all possible choices of $\mcZ$. Using the same notation, 
\begin{equation}
\label{multipermutation.eq.ballsize_fixedatequality}
\begin{split}
 \frac{\binom{\binom{nm -n}{2} + \binom{n}{2}}{\Ke} }{|\Se(\Ke)|} & \frac{(2dr+r)^n n!}{2^{2dr}n^n (r!)^m} \leq \EX[\ballsize_d(\LCMM(\bfr,\mcE))] \\ &\quad\leq \frac{\binom{\binom{nm -n}{2} + \binom{n}{2}}{\Ke} }{|\Se(\Ke)|} \frac{[(2dr+r)!]^{\frac{n}{2dr+r}}}{(r!)^m}
\end{split}
\end{equation}
where the expectation is taken over all possible choices of $\mcE$. 
\end{proposition}
\begin{proof}
See Appendix~\ref{multipermutation.appendix.proof.average_ball}.
\end{proof}

\section{Channel model and LP decoding}
\label{multipermutation.section.lpdecoding}
In the previous section we showed how to construct codes by placing linear constraints on multipermutation matrices. Recall that in Theorem~\ref{multipermutation.theorem.multiperm_polytope} we characterized the convex hull of multipermutation matrices. We now leverage this characterization to develop two linear programming decoding problems. By relaxing the ML decoding integer program, we first formulate a linear program decoding problem that is suitable for arbitrary memoryless channels.  The objective function of this LP is based on log-likelihood ratios, and is analogous to LP decoding of non-binary low-density parity-check (LDPC) codes, which is introduced by Flanagan \emph{et al.} in~\cite{flanagan2009linearprogramming}. If we apply this formulation to the AWGN channel, the resulting problem is analogous to the one developed in~\cite{wadayama2012lpdecodable}. The second problem we introduce is not seen in the literature to the best of our knowledge, and can be applied to channels that are not memoryless. In this problem, we relax the minimum Chebyshev distance decoding problem to a linear program by introducing an auxiliary variable. 

\subsection{LP decoding for memoryless channels}
We first focus on memoryless channels where $\cout$ is the channel output space. Since the initial vector $\bft$ is assumed to contain distinct entries, the channel input space is $\cin = \{t_1,\dots,t_m\}$. Without loss of generality, we assume that $t_1 < t_2 <\dots<t_m$. Let $\mpx$ be a codeword from an LP-decodable multipermutation code that is transmitted over a memoryless channel. Let $\bfy$ be the received word. Then, $\Pr_{\cout^n|\cin^n}(\bfy|\mpx) = \prod_{i = 1}^n \Pr_{\cout|\cin}(y_i|\mpxent_i)$. For this channel model, we define a function $\LLRvect: \cout \mapsto \real^m$, where $\LLRvect(y)$ is a length-$m$ row vector defined by
$\LLRent_{i}(y)= \log \left(\frac{1}{\Pr_{\cout|\cin}(y|t_i)}\right)$. Further, we let $\LLRbig(\bfy) = (\LLRvect(y_1)^T|\dots|\LLRvect(y_n)^T)^T \in \real^{mn}$.

Then, ML decoding can be written as
\begin{align*}
\hat{\mpx} &= \argmax_{\mpx \in \LCMC(\bfr,\bfA,\bfb,\trianglelefteq,\bft)} \Pr_{\cout|\cin}(\bfy|\mpx) \\
&= \argmax_{\mpx \in \LCMC(\bfr,\bfA,\bfb,\trianglelefteq,\bft)} \sum_{i = 1}^n \log \Pr_{\cout|\cin}(y_i|\mpxent_i)\\
&\eq^{\text{\scriptsize{(a)}}} \bft \left(\argmin_{\mpmX \in \LCMM(\bfr,\bfA,\bfb,\trianglelefteq)} \sum_{i = 1}^n \LLRvect(y_i) \mpmX^C_i\right) \\
&\eq^{\text{\scriptsize{(b)}}} \bft \left(\argmin_{\mpmX \in \LCMM(\bfr,\bfA,\bfb,\trianglelefteq)} \LLRbig(\bfy) \vect(\mpmX)\right),
\end{align*}
where $\mpmX^C_i$ is the $i$-th column of $\mpmX$ and the transmitted codeword is $\mpx = \bft\mpmX$. We recall that since $\mpmX$ is a multipermutation matrix, $\mpmX^C_i$ is a binary column vector with a single non-zero entry. Equality (a) comes from the fact that for each $\mpx \in \LCMC(\bfr,\bfA,\bfb,\trianglelefteq,\bft)$ there exists an $\mpmX \in \LCMM(\bfr,\bfA,\bfb,\trianglelefteq)$ such that $\mpx = \bft \mpmX$. Further, since $\LLRvect(y_i) \mpmX^C_i =- \log \left(\Pr_{\cout|\cin}(y_i|t_i)\right)$, the maximization problem can be transformed to a minimization problem.
Equality (b) is simply a change of notation. 

For this problem, we can relax the integer constraints $\LCMM(\bfr,\bfA,\bfb,\trianglelefteq)$ to linear constraints $\mpcp(\bfr,\bfA,\bfb,\trianglelefteq)$. Then the LP decoding problem is
\begin{equation}
\label{multipermutation.eq.LP_memoryless}
\begin{split}
\opmin \quad &  \LLRbig(\bfy) \vect(\bfX)\\ 
\st \quad & \bfX \in \mpcp(\bfr, \bfA,\bfb,\trianglelefteq)
\end{split}
\end{equation}

\begin{theorem}
\label{multipermutation.theorem.ml}
The LP decoding problem~\eqref{multipermutation.eq.LP_memoryless} has an ML certificate. That is, whenever LP decoding~\eqref{multipermutation.eq.LP_memoryless} outputs an integral solution, it is the ML solution.
\end{theorem}
\begin{proof}
Suppose that $\mpmX$ is the solution of the LP decoding problem and is integral. Then $\mpmX$ is a multipermutation matrix and $\bfA \vect(\mpmX) \trianglelefteq \bfb$. Since the relaxation $\mpcp(\bfr, \bfA,\bfb,\trianglelefteq)$ does not add or remove integral vertices, $\mpmX \in \LCMM(\bfr,\bfA,\bfb,\trianglelefteq)$. Since $\mpmX$ attains the maximum of the ML decoding objective, it is the ML solution.
\end{proof}

\begin{proposition}
\label{multipermutation.proposition.lp_equiv_ml}
LP decoding~\eqref{multipermutation.eq.LP_memoryless} is equivalent to ML decoding for LP-decodable multipermutation codes defined by fixed-at-zero constraints (cf. Definition~\ref{multipermutation.def.fixed-at-something}).
\end{proposition}
\begin{proof}
As before for simplicity, we denote by $\mpcp(\bfr,\mcZ)$ the code polytope of a multipermutation code subject to only fixed-at-zero constraints. In order to prove the proposition, it is sufficient to show that
\begin{equation}
\label{multipermutation.eq.fixed_at_zero_proof}
\mpcp(\bfr,\mcZ) = \conv(\LCMM(\bfr,\mcZ)).
\end{equation}
If~\eqref{multipermutation.eq.fixed_at_zero_proof} holds, then the relaxation does not have factional vertices and hence is tight. By the ML certificate (Theorem~\ref{multipermutation.theorem.ml}), LP decoding is thus equivalent to ML decoding. Note that it is easy to verify by Definition~\ref{multipermutation.def.code_polytope} that $$\mpcp(\bfr,\mcZ) \supset \conv(\LCMM(\bfr,\mcZ)).$$ Hence, to complete the proof, we need to show that for all $\bfZ \in \mpcp(\bfr,\mcZ)$, $\bfZ \in \conv(\LCMM(\bfr,\mcZ))$. 

Since $\bfZ \in \mpcp(\bfr,\mcZ)$, we can express $\bfZ$ as a convex combination of multipermutation matrices in $\mpmset(\bfr)$. In other words,  
\begin{align*}
\bfZ =  \sum_{h = 1}^{|\mpmset(\bfr)|}\alpha_h \mpmX_h,
\end{align*}
where $\mpmX_h \in \mpmset(\bfr)$ are multipermutation matrices, and the set $\{\alpha_h\}$ is a set of convex combination coefficients. We split the sum to two parts:
\begin{align*}
\bfZ =  \sum_{h: \mpmX_h \in \LCMM(\bfr,\mcZ)}\alpha_h \mpmX_h + \sum_{h: \mpmX_h \notin \LCMM(\bfr,\mcZ)}\alpha_h \mpmX_h.
\end{align*}
Since $Z_{ij} = 0$ for all $(i,j) \in \mcZ$, $\alpha_h = 0$ for all $h$ such that $\mpmXent_{ij,h} \neq 0$. This means that $\alpha_h = 0$ for all $\mpmX_h \notin \LCMM(\bfr,\mcZ)$, which implies that 
\begin{align*}
\bfZ =  \sum_{h: \mpmX_h \in \LCMM(\bfr,\mcZ)}\alpha_h \mpmX_h.
\end{align*}
This implies that $\bfZ \in \conv(\LCMM(\bfr,\mcZ))$.
\end{proof}
\subsubsection{The AWGN channel}
In the AWGN channel, $\Pr_{\cout|\cin}(y|t_i) = \frac{1}{\sqrt{2\pi}\sigma}e^{\frac{(y - t_i)^2}{2\sigma^2}}$, where $\sigma^2$ is the variance of the noise. Thus 
\begin{align*}
\LLRbig(\bfy) = &\phi\cdot\bfone_{1\times mn} + \Kc(\underbrace{(y_1 - t_1)^2,\dots,(y_1 - t_m)^2}_{m}\\&\qquad\qquad\dots \underbrace{(y_n - t_1)^2, \dots,(y_n - t_m)^2}_{m}),
\end{align*}
 where $\phi = \log\frac{1}{\sqrt{2\pi}\sigma}$ is a constant bias and $\Kc = \frac{1}{2\sigma^2} >0$ is a common scaling constant. Then $$\LLRbig(\bfy) \vect(\mpmX) = n\phi + \Kc\left(\sum_{i = 1}^n y_i^2 + \sum_{i = 1}^m r_i t_i^2 - 2 \bfu \vect(\mpmX)\right),$$ where $\bfu = (\underbrace{y_1 t_1,\dots, y_1 t_m}_{m}\dots\underbrace{y_n t_1, \dots, y_n t_m}_{m})$. Thus 
\begin{equation}
\argmin_{\bfX} \LLRbig(\bfy) \vect(\bfX) = \argmax_{\bfX} \tr((\bfy^T \bft)\bfX).
\end{equation}

We note that when the multiplicity vector is the all-ones vector, this formulation is the same as the LP decoding problem proposed in~\cite{wadayama2012lpdecodable}. As a result, it is easy to restate the definition of pseudodistance and the error bound properties in~\cite{wadayama2012lpdecodable} for LP-decodable multipermutation codes. We refer readers to~\cite[Section~IV]{wadayama2012lpdecodable} for details.
\subsubsection{Discrete memoryless $q$-ary symmetric channel}
For this channel, the channel output space is the same as the input space. Namely, $\cin = \cout = \{t_1,\dots,t_m\}$. The transition probabilities are given by
\begin{equation*}
\Pr_{\cout|\cin}(y|x) = \begin{cases}
1 -p & \text{ if } y = x\\
\frac{p}{m - 1} & \text{ otherwise.}
\end{cases}
\end{equation*}
Let $\bfe(y)$ be a row vector such that $e_i(y) = 0$ if $y\neq t_i $ and $e_i(y) = 1$ if $y = t_i$. Further, we denote by $\bfY$ the matrix
\begin{equation*}
\bfY = [\bfe(y_1)^T | \bfe(y_2)^T | \dots | \bfe(y_n)^T].
\end{equation*}
Using this notation,
\begin{equation*}
\LLRvect(\bfy)  = \log\left(\frac{m-1}{p}\right) \bfone +  \log\left(\frac{1}{1-p}\cdot\frac{p}{m-1}\right)\bfe(y_i).
\end{equation*}
Then, 
\begin{align*}
\LLRbig(\bfy) \vect(\mpmX) = \tr&\left( \log\left(\frac{m-1}{p}\right)\bfE^T \mpmX \right.\\ &\left. +\log\left(\frac{1}{1-p}\cdot\frac{p}{m-1}\right) \bfY^T \mpmX \right),
\end{align*}
where $\bfE$ is an $m\times n$ matrix with all entries equal to one. Note that $\tr(\bfE^T \mpmX) = n$ is a constant and $\log\left(\frac{1}{1-p}\cdot\frac{p}{m-1}\right)$ is a negative constant. Therefore 
\begin{equation}
\argmin_{\bfX} \LLRbig(\bfy) \vect(\bfX) = \argmax_{\bfX} \tr(\bfY^T\bfX).
\end{equation}
Note that this is equivalent to minimizing the Hamming distance between $\bfX$ and $\bfY$ (cf. Lemma~\ref{multipermutation.lemma.Hamming}).

\subsection{LP decoding for the Chebyshev distance}
In this subsection we relax the problem of minimum Chebyshev distance decoding to a linear program. Minimum Chebyshev distance decoding can be written as the following optimization:
\begin{align*}
\opmin \quad & \max_{i} |\mpxent_i - \mpyent_i| \\
\st \quad &  \mpx \in \LCMC(\bfr,\bfA,\bfb,\trianglelefteq,\bft).
\end{align*}
We introduce an auxiliary variable $\delta$ and rewrite the problem as
\begin{align*}
\opmin \quad & \delta  \\ \st \quad &  \mpx \in \LCMC(\bfr,\bfA,\bfb,\trianglelefteq,\bft),\\
	& -\delta \leq \mpxent_i - \mpyent_i \leq \delta\text{ for all }i.
\end{align*}
Note that $\mpx = \bft \mpmX$, where $\mpmX \in \LCMM(\bfr,\bfA,\bfb,\trianglelefteq)$. Therefore the problem can be reformulated as
\begin{align*}
\opmin \quad &  \delta  \\
\st \quad &  \mpmX \in \LCMM(\bfr,\bfA,\bfb,\trianglelefteq),\\
	& -\bfdelta \leq \bft \mpmX - \mpy \leq \bfdelta,
\end{align*}
where $\bfdelta := (\delta, \delta,\dots,\delta)$ is a length-$n$ vector. To relax the problem to an LP, we replace $\LCMM(\bfr,\bfA,\bfb,\trianglelefteq)$ by $\mpcp(\bfr, \bfA,\bfb,\trianglelefteq)$ and obtain
\begin{equation}
\label{multipermutation.eq.lp_chebyshev}
\begin{split}
\opmin \quad  & \delta  \\
\st \quad &  \mpmX \in \mpcp(\bfr, \bfA,\bfb,\trianglelefteq), \\ 
	& -\bfdelta \leq \bft \mpmX - \mpy \leq \bfdelta.
\end{split}
\end{equation}

We make two remarks. First, as already mentioned, due to the relaxation, optimizer of the LP decoding problem may contain fractional entries. When this is the case, the decoding should be considered to be a decoding failure. However, it is not hard to observe that we can round the results in hope of finding ML solution. In this paper, we adopt a simple rounding heuristic to obtain the final decoding result. Let $\hat{x}_j = t_{\hat{i}(j)}$, where $ \hat{i}(j) = \argmax_i X_{ij}$ for all $j = 1,\dots,n$. Note that this step is important for LP decoding of Chebyshev distance since the solution to~\eqref{multipermutation.eq.lp_chebyshev} is empirically observed to contain many fractional entries. 

Second, both LP decoding formulations can be solved using off-the-shelf solver such as the CVX toolbox~\cite{cvx}. However, generic LP solvers do not automatically exploit the structure of the LP decoding problem. In particular, the constraints in Theorem~\ref{multipermutation.theorem.multiperm_polytope} can be described using factor graphs. We leverage this insight in Section~\ref{multipermutation.subsection.admm} to develop an efficient decoding algorithm.

\section{Encoding and decoding algorithms for LP-decodable multipermutation codes}
\label{multipermutation.section.algorithms}
To make our previous contributions more practical, in this section we focus on encoding and decoding algorithms primarily for ST codes. However, note that the decoding algorithm we develop can be generalized to decode all LP-decodable multipermutation codes. To the best of our knowledge, there have been no encoding nor decoding algorithms developed for ST codes. We note that it is simple to derive a bounded distance decoder for ST codes by extending the decoding method proposed for Construction~1 in~\cite{tamo2010correcting}. However, we have not been successful in finding an encoding algorithm in the literature. Therefore, we first introduce a method that encodes ST codes, and then develop an efficient ADMM algorithm for the LP decoding problem~\eqref{multipermutation.eq.LP_memoryless}.
\subsection{An encoding algorithm for ST codes}
\label{multipermutation.subsection.encoding}
Formally, the encoding task for ST codes is as follows: Given a message from $\{0,\dots,|\codebook_{ST}| - 1\}$, where $\codebook_{ST}$ is the codebook and $|\codebook_{ST}|$ is the cardinality of the codebook, the algorithm should map the message index to the corresponding codeword of $\codebook_{ST}$. 

Before proceeding to the encoding algorithm, we first present mapping between the $N = \frac{(\sum_{i = 1}^m r_i)!}{\prod_{i = 1}^m (r_i!)}$ multipermutations and the integers from $\{0,\dots,N-1\}$. Denote by ``\texttt{rankMP}()'' the map from multipermutation to integer, and denote by ``\texttt{unrankMP}()'' the inverse map. We only describe \texttt{rankMP}() because it is straightforward, and because \texttt{unrankMP}() can be deduced from \texttt{rankMP}(). To the best of our knowledge, the only previous such mapping is an unpublished online posting due to \v{S}avara in~\cite{savara2007ranking}. The mapping in~\cite{savara2007ranking} ranks multipermutations in lexicographical order. Our algorithm produces a different ordering based on a novel mixed radix number system interpretation of multipermutations.

We summarize \texttt{rankMP}($\mpx$) in Algorithm~\ref{multipermutation.algorithm.rank_mp}, where $\mpx$ is a multipermutation parameterized by multiplicity vector $\bfr$. The intuition of the algorithm is as follows. Multipermutations can be considered as a mixed radix number system that has $m$ ``digits''. Each digit has a ``base'' that is the total number of induced combinations within the multipermutation (cf. Step~\ref{multipermutation.algorithm.step.base}). The digits themselves can be calculated using~\cite[Theorem L]{knuth2005combinations}, which maps combinations to integers  (Step~\ref{multipermutation.algorithm.step.combinatoric}). This process is demonstrated in Example~\ref{multipermutation.example.encoding_algorithms}.

\begin{algorithm}
\caption{$M =$ \texttt{rankMP}$(\mpx)$}
\label{multipermutation.algorithm.rank_mp}
\begin{algorithmic}[1]
\STATE $\mpy \leftarrow \mpx$, let $n_y$ be the length of $\mpy$. The $1$-st base is always $1$, i.e., $b_1 = 1$.
\FORALL{ $i = 1,\dots,m$}
\STATE Construct the vector $(\alpha_1,\dots,\alpha_{r_i})$ such that  $y_{\alpha_j + 1} = i$ for all $j = 1\dots,r_i$.
\STATE \label{multipermutation.algorithm.step.combinatoric} Calculate the $i$-th digit, $a_i \leftarrow \sum_{j = 1}^{r_i} \binom{\alpha_j} {j}$. 
\STATE \label{multipermutation.algorithm.step.base} Calculate the $(i + 1)$-th base, $b_{i+1} \leftarrow \binom{n_y}{r_i}$.
\STATE Update $\mpy$ by deleting $y_{\alpha_j}, \forall j = 1\dots,r_i$. Update $n_y$.
\ENDFOR
\STATE $M = \sum_{i = 1}^m a_i b_i$.
\end{algorithmic}
\end{algorithm}

As mentioned above, one can invert Algorithm~\ref{multipermutation.algorithm.rank_mp} and obtain \texttt{unrankMP}(). Note that once $\bfr$ is fixed, the bases $b_i$ are fixed. As a result, in order to invert \texttt{unrankMP}() one should use modular arithmetic to determine $a_i$, and then invert Step~\ref{multipermutation.algorithm.step.combinatoric} again using modular arithmetic. We omit the details but demonstrate this process in Example~\ref{multipermutation.example.encoding_algorithms}.

Based on \texttt{unrankMP}(), we now develop an algorithm that encodes ST codes. Let $\mpx$ be a codeword. Then, by Definition~\ref{multipermutation.def.chebyshev_code}, $x_i \equiv i \bmod{d}$ for all $i$. We split $\mpx$ into $d$ sub-vectors, $\mpx^{(1)},\dots,\mpx^{(d)}$, such that $\mpx^{(1)} = (x_{1}, x_{d + 1},x_{2d + 1},\dots)$, $\mpx^{(2)} = (x_{2}, x_{d + 2},x_{2d + 2},\dots)$, and so on. As a result, $\mpx^{(k)}$ is a $r$-regular multipermutation that (multi-)permutes the initial vector $\bft^{(k)} = (k,d+k,2d+k,\dots)$. This means that one can encode each $\mpx^{(k)}$, $k = 1,\dots,d$ independently. We summarize this idea in Algorithm~\ref{multipermutation.algorithm.encode_st_codes}.

\begin{algorithm}
\caption{$\mpx =$ \texttt{encodeST}$(M)$}
\label{multipermutation.algorithm.encode_st_codes}
\begin{algorithmic}[1]
\STATE Consider the number system of radix $\frac{(ar)!}{(r!)^a}$, where $a = m/d$. Convert the integer $M$ to a vector of digits denoted as $(l_1,\dots,l_d)$.
\FORALL{ $k = 1,\dots,d$}
\STATE Do \texttt{unrankMP}($l_k$) and obtain the $r$-regular multipermutation of length $n/d$. 
\STATE Apply this multipermutation to the initial vector $\bft^{(k)} = (k,d+k,2d+k,\dots)$ to obtain $\mpx^{(k)}$.
\ENDFOR
\STATE Combine $\mpx^{(k)}$ for all $k = 1,\dots,d$ by merging entries. 
\end{algorithmic}
\end{algorithm}
\begin{example}
\label{multipermutation.example.encoding_algorithms}
We demonstrate Algorithm~\ref{multipermutation.algorithm.rank_mp} and~\ref{multipermutation.algorithm.encode_st_codes} via examples. 

We first calculate \texttt{rankMP}$(3,3,2, 1, 1, 2)$. At the beginning, $n_y = 6$ and $b_1 = 1$. We obtain $(\alpha_1, \alpha_2) = (3,4)$, and thus $a_1 = \binom{3}{1} + \binom{4}{2} = 9$. Furthermore, $b_2 =  \binom{6}{2} = 15$. By deleting $y_4$ and $y_5$, we obtain the updated $\bfy' = (3, 3, 2, 2)$. Continuing the previous process, we get $(\alpha_1', \alpha_2') = (2,3)$ and $a_2 = \binom{2}{1} + \binom{3}{2} = 5$. As a result, $M = 9 \cdot 1 + 5\cdot 15 = 84$. In this example, $84$ is expressed by two digits: $5_{15}9_1$. Each digit belongs to a different base. 

The inverse algorithm, i.e., \texttt{unrankMP}($84$), requires knowledge of the multiplicity vector $\bfr = (2,2,2)$. The bases are easy to determine: $b_1 = 1$ and $b_2 = \binom{6}{2}$. As a result, $a_1 = 9$ because $9 \equiv 84 \bmod{b_2}$. Further, $a_2 = 5$ since $a_2 b_2 + a_1 b_1 = 84$. We then recover $\alpha_2' = 3$, which is the largest integer such that $\binom{\alpha_2'}{2}\leq 5$. As a result, by solving $\binom{\alpha_1'}{1} = 5 - \binom{3}{2}$, we obtain $\alpha_1' = 2$. Repeating this process, we recover $(\alpha_1, \alpha_2) = (3,4)$. Finally, using $(\alpha_1, \alpha_2)$ and $(\alpha_1', \alpha_2')$, which are vectors that describe the position of each value, we reconstruct the multipermutation as $\bfy = (3,3,2,1,1,2)$.

Next, we use Algorithm~\ref{multipermutation.algorithm.encode_st_codes} to encode a message. Consider the ST code with parameters $r = 2$, $d = 3$, and $m = 6$. This code is of cardinality $216$. Suppose we would like to encode message $137$. By converting $137$ to a vector of digits with base $6$, we first obtain $(l_1,l_2,l_3) = (3, 4, 5)$, i.e., $3\cdot 6^2 + 4\cdot 6 + 5 = 137$. For each digit $l_k$, we use \texttt{unrankMP}($l_k$) to calculate the corresponding $2$-regular multipermutation of length $4$. The results in this step are $3 \rightarrow (1,2,2,1)$, $4 \rightarrow (2,1,2,1)$, and $5 \rightarrow (2,2,1,1)$. Furthermore, we obtain $\bft^{(1)}= (1,4)$, $\bft^{(2)}= (2,5)$, and $\bft^{(3)}= (3,6)$. Combining these results, we obtain $\mpx^{(1)} = (1,4,4,1)$, $\mpx^{(2)} = (5,2,5,2)$, and $\mpx^{(3)} = (6,6,3,3)$. Therefore the codeword for the message $137$ is $\mpx = (1,5,6,4,2,6,4,5,3,1,2,3)$.
\end{example}
\subsection{LP decoding of linearly constrained multipermutation codes using the alternating direction method of multipliers (ADMM)}
\label{multipermutation.subsection.admm}
In this subsection, we formulate the LP decoding problem~\eqref{multipermutation.eq.LP_memoryless} as an instance of ADMM. We first introduce a factor graph representation for codes constrained by fixed-at-zero and/or fixed-at-equality constraints. Then, we reformulated the decoding problem in the template of ADMM.
\subsubsection{Factor graph representation}
\label{multipermutation.subsection.factor_graph}
By Definition~\ref{multipermutation.def.multipermutation_matrix}, a multipermutation matrix is a binary matrix satisfying $m$ row sum constraints and $n$ column sum constraints. This constraint satisfaction problem can be represented using $mn$ variable nodes, $m$ row-sum-check nodes, and $n$ column-sum-check nodes. It can be drawn as a graph with circles representing variables nodes, squares representing row-sum-check nodes, and triangles representing column-sum-check nodes. 

When additional constraints are enforced by Definition~\ref{multipermutation.def.linearly_constrained_multipermutation_matrix}, the factor graph needs to be modified to reflect these added constraints. In particular, if fixed-at-zero constraints are used, then we delete all variable nodes that correspond to entries $(i,j) \in \mcZ$. On the other hand, if fixed-at-equality constraints are used, then for each pair $(i,j),(k,l) \in \mcE$, we delete node $(k,l)$ and reconnect edges originally connected to $(k,l)$, so that they are connected to $(i,j)$. We illustrate this process in Example~\ref{multipermutation.example.factorgraph}.
\begin{example}
\label{multipermutation.example.factorgraph}
Consider multipermutation matrices parameterized by $\bfr = (1,2,1)$. In Figure~\ref{multipermutation.fig.mp_matrix}, we draw the factor graph for the set of all multipermutation matrices. If, in addition, $X_{31} = 0$ and $X_{13} = X_{24}$, then we delete node $(3,1)$ and $(2,4)$, and modify the edges originally connected to node $(2,4)$ so that they are connected to node $(1,3)$. The resulting graph is showed in Figure~\ref{multipermutation.fig.modified_mp_matrix}. 
\begin{figure}
\psfrag{&Row}{\scriptsize{$\sum_{i = 1}^3 X_{ij} = 1$}}
\psfrag{&c1}{\scriptsize{$\sum_{j = 1}^4 X_{1j} = 1$}}
\psfrag{&c2}{\scriptsize{$\sum_{j = 1}^4 X_{2j} = 2$}}
\psfrag{&c3}{\scriptsize{$\sum_{j = 1}^4 X_{3j} = 1$}}
	\begin{center}
    \includegraphics[width= 16pc]{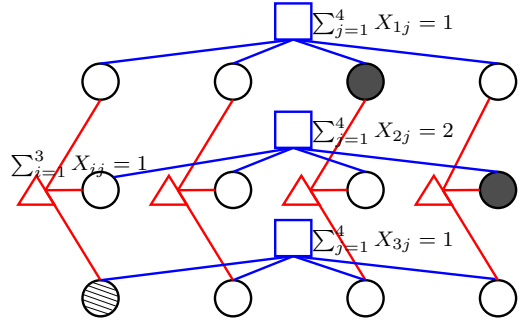}
    \end{center}
    \caption{Factor graph of multipermutation matrices parameterized by $\bfr = (1,2,1)$. $X_{31}$, $X_{13}$, and $X_{24}$ are highlighted.}
    \label{multipermutation.fig.mp_matrix}
\end{figure}

\begin{figure}
\psfrag{&r2}{\scriptsize{$X_{21} + X_{22} + X_{23} + X_{13} = 2$}}
\psfrag{&c1}{\scriptsize{$X_{11} + X_{21} = 1$}}
\psfrag{&c4}{\scriptsize{$X_{13} + X_{14} + X_{34}  = 1$}}
	\begin{center}
    \includegraphics[width= 16pc]{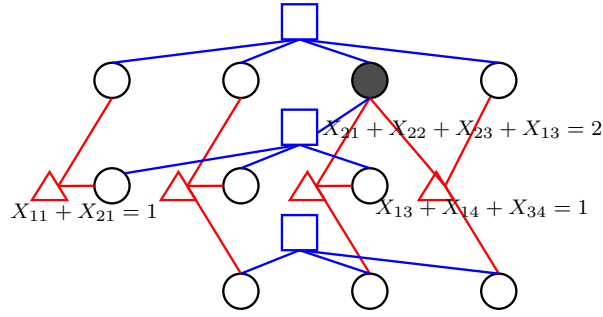}
    \end{center}
    \caption{Factor graph of multipermutation matrices parameterized by $\bfr = (1,2,1)$. In addition, $X_{31} = 0$ and $X_{13} = X_{24}$.}
    \label{multipermutation.fig.modified_mp_matrix}
\end{figure}
\end{example}

\subsubsection{ADMM algorithm for the LP decoding problem~\eqref{multipermutation.eq.LP_memoryless}}
ADMM based LP decoding of binary LDPC codes is introduced in~\cite{barman2013decomposition}. The ideas developed in~\cite{barman2013decomposition} motivate us to develop an ADMM algorithm for decoding LP-decodable multipermutation codes. We first introduce some notation that is useful in deriving the algorithm, and then state the ADMM formulation.

First, for compactness, we use $\bfgamma$ and $\optbfx$ to represent $\LLRbig(\bfy)$ and $\vect(\bfX)$ respectively. Next, we introduce selection matrices $\bfPc_j$, $j = 1,\dots,n$, such that $\bfPc_j \optbfx$ selects entries from $\optbfx$ that participate in the $j$-th column-sum-check. Similarly, let $\bfPr_i$, $i = 1,\dots,m$, be selection matrices, each of which selects entries from $\optbfx$ that participate in the corresponding row-sum-check. Finally, we denote by $\simplex_m$ the standard $m$-simplex, i.e., the polytope defined by 
$\simplex_m = \{(x_1,\dots,x_{m})\in \real^{m} |
  \sum_{k = 1}^{m} x_k = 1, \text{ and }x_k \geq 0 \text{ for all }k\}$. Furthermore, let $\oneslice_{n}^{r}$ be defined as follows: 
$\oneslice_{n}^{r} = \{(x_1,\dots,x_{n})\in \real^{n} |
  \sum_{k = 1}^{n} x_k = r, \text{ and }1 \geq x_k \geq 0 \text{ for all }k\}.$

Equipped with the notation above, we rewrite~\eqref{multipermutation.eq.LP_memoryless} as
\begin{equation}
\label{multipermutation.eq.LP_memoryless_reform}
\begin{split}
\opmin \quad &  \bfgamma^T \optbfx\\ 
\st \quad & \bfPr_i \optbfx \in \oneslice_{n}^{r_i}, \forall i = 1,\dots,m,\\
& \bfPc_j \optbfx \in \simplex_m, \forall j = 1,\dots,n,\\
&   \bfA \optbfx \trianglelefteq \bfb.
\end{split}
\end{equation}

The next step in ADMM is the exploit the structure of the constraint $\bfA \optbfx \trianglelefteq \bfb$. For example, when the code is only constrained by fixed-at-zero and fixed-at-equality constraints, the constraints in~\eqref{multipermutation.eq.LP_memoryless_reform} can be translated to a modified factor graph as discussed in Section~\ref{multipermutation.subsection.factor_graph}. This translation results in the following changes to~\eqref{multipermutation.eq.LP_memoryless_reform}. First, the $\bfPc_j$'s and $\bfPr_i$'s should be changed to match the modified factor graph. Then, the parameters of $\oneslice$ and $\simplex$ should be revised accordingly. Finally, ``$\bfA \vect(\bfX) \trianglelefteq \bfb$'' can be removed since the two steps above are sufficient to describe this type of constraint set.

For simplicity, we use an ST code with parameters $r$, $d$, and $m$ to illustrate the ADMM based decoding algorithm. Note that this formulation is easily extended to other codes with fixed-at-zero and fixed-at-equality constraints. Using the techniques developed in~\cite{barman2013decomposition}, we introduce replicas $\optbfzc$ and $\optbfzr$ to rewrite problem~\eqref{multipermutation.eq.LP_memoryless_reform} as
\begin{equation}
\label{multipermutation.eq.LP_memoryless_replica}
\begin{split}
\opmin \quad &  \bfgamma^T \optbfx\\ 
\st \quad & \bfPr_i \optbfx = \optbfzr_i, \bfPc_j \optbfx = \optbfzc_j,\\
& \optbfzr_i \in \oneslice_{n/d}^{r}, \optbfzc_j \in \simplex_{m/d}.
\end{split}
\end{equation}
Then, the augmented Lagrangian used in ADMM is
\begin{align*}
\Lag_{\mu} (\optbfx, \optbfzr, &\optbfzc, \bflambda, \bfeta) = \bfgamma^T \optbfx 
\\&+\sum_{j}\bflambda^T_{j} (\bfPc_j \optbfx - \optbfzc_j) + \frac{\mu}{2}\sum_{j}\|\bfPc_j \bfx - \optbfzc_j\|_2^2 
\\&+\sum_{i}\bfeta^T_{i} (\bfPr_i \optbfx - \optbfzr_i) + \frac{\mu}{2}\sum_{i}\|\bfPr_i \optbfx - \optbfzr_i\|_2^2.
\end{align*}
The ADMM algorithm minimizes $\Lag_{\mu} (\optbfx, \optbfzr, \optbfzc, \bflambda, \bfeta)$ in an iterative fashion similar to the one in~\cite{barman2013decomposition}, and hence we omit the details. Instead, we make four remarks.
\begin{enumerate}
\item Although there are two symbols for replicas, $\optbfzr$ and $\optbfzc$, one can concatenate $\optbfzr$ and $\optbfzc$ to form one vector. Doing so does not change the algorithm.
\item The resulting $\optbfx$-update step (cf.~\cite{barman2013decomposition}) for~\eqref{multipermutation.eq.LP_memoryless_replica} is an average of the corresponding replicas plus a bias from $\bfgamma$.
\item The resulting $\optbfz$-update step (cf.~\cite{barman2013decomposition}) requires two types of projections: projection onto $\simplex_m$ and projection onto $\oneslice_{n}^{r}$. The first projection can be solved in linear time using techniques developed in~\cite{duchi2008efficient}. The second projection can be solved in linear time using Algorithm~\ref{multipermutation.algorithm.linear_time_projection} presented in Appendix~\ref{multipermutation.appendix.linear_time_projection}, and which is developed based on ideas first proposed in~\cite{gupta2010l1}. The details of this linear time projection algorithm are presented in Appendix~\ref{multipermutation.appendix.linear_time_projection}.
\item Each ADMM iteration consists of an $O(mn)$ $\bfx$-update, $m$ projections onto $\oneslice_{n}^{r}$, $n$ projections onto $\simplex_m$, and an $O(mn)$ $\lambda$-update. Therefore, the computational complexity per iteration is $O(mn)$. Due to the convergence result in~\cite[Proposition~1]{barman2013decomposition}, the ADMM algorithm in this section has time complexity on the order of $mn$.
\end{enumerate}

We note that ADMM can also decode LP-decodable permutation codes, e.g., the pure involution code introduced in~\cite{wadayama2012lpdecodable}. Further, as demonstrated in~\cite{barman2013decomposition}, ADMM can decode long block length LDPC codes efficiently. Therefore a promising future work is to design long block length multipermutation codes with sparse structure.

\section{Numerical results}
In this section we present simulation results for ST codes with various parameter settings. We simulate the AWGN channel following the methodology adopted in~\cite{wadayama2012lpdecodable} and~\cite{zhang2010ldpc}.
We present results comparing five classes of decoders as follows:
\begin{itemize}
\item LP decoding~\eqref{multipermutation.eq.LP_memoryless} (denoted by ``LP AWGN''). Note that by Proposition~\ref{multipermutation.proposition.lp_equiv_ml}, LP decoding is equivalent to ML decoding for ST codes. We also verify this claim empirically by implementing ML decoding via an exhaustive search (denoted by ``ML AWGN'').
\item Minimum distance decoding (denoted by ``Minimum distance''). We first rank the channel outputs to form a multipermutation that has the same multiplicity vector as the codebook. Then, we minimize the Chebyshev distance via an exhaustive search over all codewords.
\item Soft LP decoding of Chebyshev distance (denoted by ``LP Chebyshev, soft''). We take the channel output as $\bfy$ in~\eqref{multipermutation.eq.lp_chebyshev} and solve~\eqref{multipermutation.eq.lp_chebyshev}. Note that $\bfy$ is a real valued vector and is not necessarily a multipermutation. The distance we minimize is the infinity norm between two real-valued vectors.
\item Hard LP decoding of Chebyshev distance (denoted by ``LP Chebyshev, hard''). We first rank the channel outputs to a multipermutation with the same multiplicity vector as the codebook. Then, we use this ranking as $\bfy$ for problem~\eqref{multipermutation.eq.lp_chebyshev}.
\item Bounded distance decoding (denoted by ``Bounded distance''). We first rank the channel outputs to a multipermutation with the same multiplicity vector as the codebook. Then, we search for the unique codeword within radius $d/2$ in the Chebyshev metric, where $d$ is the minimum Chebyshev distance of the ST code. We declare an error if no codeword is found or, more than one codeword is found. As mentioned in Section~\ref{multipermutation.section.algorithms}, by extending the decoding method for Construction 1 introduced in~\cite{tamo2010correcting}, one can construct an efficient bounded distance decoder for ST codes that corrects $(d-1)/2$ errors when $d$ is odd.
\end{itemize}

We first simulate the ST code with parameters $r = 2$, $d = 3$, and $m = 6$. This means that each codeword is of length $12$ and that bounded distance decoding can correct $1$ error in the Chebyshev metric. In Figure~\ref{multipermutation.fig.STCodes_123456123456}, we plot the word-error-rate (WER) as a function of signal-to-noise ratio (SNR)\footnote{SNR is defined by $10\log_{10}\frac{1}{\sigma^2}$, where $\sigma^2$ is the variance of the Gaussian noise.} for the codeword $(1,2,3,4,5,6,1,2,3,4,5,6)$, where each data point is based on $100$ word errors. As this is a nonlinear code, it is important to note that the WER performance is not necessarily the same across different codewords. Since the code is proven to correct some number of errors, our intention here is to demonstrate that LP decoding performs much better than bounded distance decoding. Nevertheless, we also simulated several other random codewords. We observed that, for the set of codewords we simulate, the WER does not vary significantly across different codewords (data not shown).

We make the following observations.
First, LP (ML) decoding achieve a significantly lower error rate than the other decoders. This suggests that soft decoding is better than hard decoding in terms of error rates. The difference between the WER plots of LP and ML decoding one may observe in Figure~\ref{multipermutation.fig.STCodes_123456123456} is due to statistical fluctuation.
Second, soft and hard LP decoding of Chebyshev distance suffer a $2$ to $4$ dB loss when compared to minimum distance decoding. In addition, these two decoders both achieve WER performance similar to that of bounded distance decoding. Because bounded distance decoding is much more computationally efficient, for this code, we prefer bounded distance decoding to LP decoding of Chebyshev distance. 
\begin{figure}
\psfrag{SNR}{\scriptsize{SNR (dB)}}
\psfrag{WER}{\hspace{-0.5cm}\scriptsize{word-error-rate}}
\psfrag{&BDDecoding}{\scriptsize{Bounded distance}}
\psfrag{&LPDecChebyshevHard-}{\scriptsize{LP Chebyshev, hard}}
\psfrag{&LPDecChebyshevSoft}{\scriptsize{LP Chebyshev, soft}}
\psfrag{&MDDecoding}{\scriptsize{Minimum distance}}
\psfrag{&LP}{\scriptsize{LP AWGN}}
\psfrag{&ML}{\scriptsize{ML AWGN}}
	\begin{center}
    \includegraphics[width = 21pc]{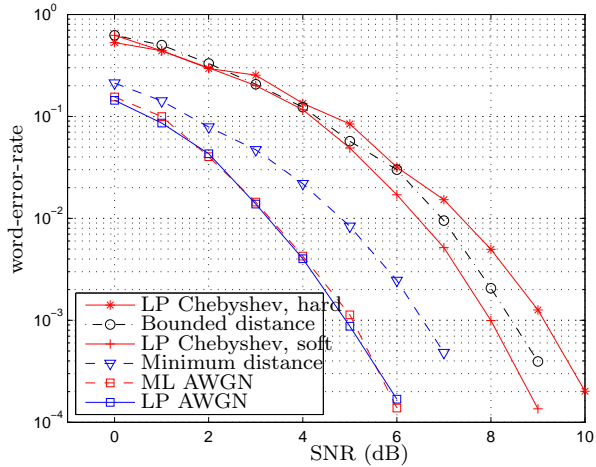}
    \end{center}
    \caption{Word-error-rate (WER) plotted as a function of signal-to-noise ratio (SNR) for the ST code with parameters $r = 2$, $d = 3$, and $m = 6$. The codeword transmitted is $(1,2,3,4,5,6,1,2,3,4,5,6)$.}
\label{multipermutation.fig.STCodes_123456123456}
\end{figure}

Next, we consider the ST code with parameters $r = 3$, $d = 4$, and $m = 16$. We plot WER as a function of SNR in Figure~\ref{multipermutation.fig.STCodes_1-161-161-16}. This code has a block length of $48$, which is larger than the first code, making exhaustive search expensive to implement. As a result, we do not present results for the exhaust search based minimum distance decoding. However, we can implement ML decoding via ADMM. In ADMM, we set $\mu= 5.5$ and the maximum number of iterations $T_{\max} = 200$.
We observe that, unlike in Figure~\ref{multipermutation.fig.STCodes_123456123456}, both soft and hard LP decoding of Chebyshev distance significantly outperform bounded distance decoding. However, their performance is about $2$ to $3$ dB worse than LP decoding~\eqref{multipermutation.eq.LP_memoryless}. From a computational complexity point of view, we note that although the maximum number of iterations is set to $200$, the average number of iterations observed was less than $50$ at all SNRs simulated. We omit the details since the behavior is similar to that for ADMM decoding of LDPC codes. We refer the reader to~\cite{barman2013decomposition} for an extensive discussion of implementation of ADMM for binary LDPC codes.
\begin{figure}
\psfrag{SNR}{\scriptsize{SNR (dB)}}
\psfrag{WER}{\hspace{-0.5cm}\scriptsize{word-error-rate}}
\psfrag{&BDDecoding}{\scriptsize{Bounded distance}}
\psfrag{&LPDecChebyshevHard-}{\scriptsize{LP Chebyshev, hard}}
\psfrag{&LPDecChebyshevSoft}{\scriptsize{LP Chebyshev, soft}}
\psfrag{&LPML}{\scriptsize{LP (ML) AWGN}}
	\begin{center}
    \includegraphics[width = 21pc]{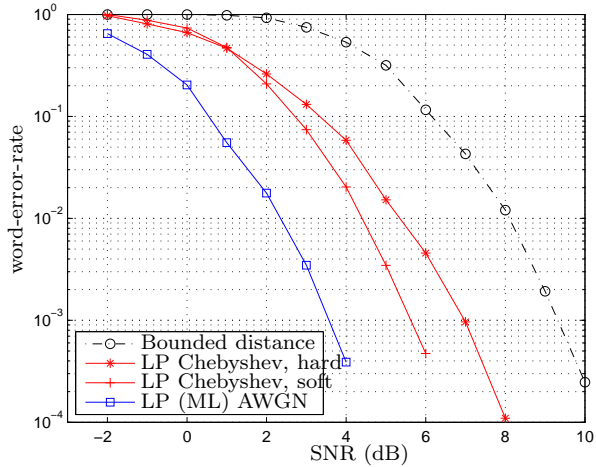}
    \end{center}
    \caption{Word-error-rate (WER) plotted as a function of signal-to-noise ratio (SNR) for the ST code with parameter $r = 3$, $d = 4$, and $m = 16$. The codeword transmitted is $(1,\dots,16,1,\dots,16,1,\dots,16)$.}
    \label{multipermutation.fig.STCodes_1-161-161-16}
\end{figure}

\section{Conclusions}
In this paper, we develop several fundamental tools of a framework for codes based on multipermutation. 

We first develop new theories: We propose representing multipermutations using binary matrices that we term multipermutation matrices. Using multipermutation matrices, we define LP-decodable multipermutation codes. In order to decode these codes using LP decoding we characterize the convex hull of multipermutation matrices, which is analogous to the Birkhoff polytope of permutation matrices. Using this result, we relax the code constraints and formulate two LP decoding problems. The first decoding problem minimizes the ML decoding objective. It can be applied to arbitrary memoryless channel. The second decoding problem minimizes the Chebyshev distance. 

To make these contributions useful in practice, we also develop new algorithms. We first develop a mixed radix number system interpretation for multipermutation. We use it to develop an efficient encoding algorithm for ST codes. Regarding decoding algorithms, we reformulate the LP decoding problem and use ADMM to solve it. The resulting ADMM formulation requires two projection subroutines that can be solved efficiently using techniques drawn from the literature. 

These contributions result in two major advantages for LP-decodable multipermutation codes. First, both LP decoding problems presented in this paper are computationally tractable. In particular, the LP decoding problem for memoryless channels can be solved efficiently using ADMM. Second, our simulation results indicate that LP decoding can achieve significantly lower error rates than hard decoding algorithms such as bounded distance decoding. 

The above two advantages lead to new research directions: The first is the design of good codes. ADMM LP decoding is simple and efficient for codes with fixed-at-zero and fixed-at-equality constraints. Consequently, code designs that use these two types of constraints can be decoded efficiently using LP decoding. In fact, we already know some codes that benefit from the algorithm, e.g., ST codes and pure involution codes.

Second, although soft decoding can achieve lower error rates than hard decoding, soft decoding requires knowledge of the initial vector. However, in many situations, such knowledge is missing. As an example, in rank modulation, the decoder does not know the exact values of the initial vector determined at the cell programming stage. It hence cannot calculate the log-likelihood ratios of~\eqref{multipermutation.eq.LP_memoryless}. Therefore, it is important to develop decoding algorithms that can deal with uncertainty in the initial vector. We present some initial ideas along these lines in Appendix~\ref{multipermutation.appendix.initial_vector}. 
\section*{Appendix}
\section{Random coding results}
\label{multipermutation.appendix.random_coding_results}
\subsection{Proofs of random coding results}
\subsubsection{Proof of Lemma~\ref{multipermutation.lemma.codebook_symmetry}}
\label{multipermutation.appendix.proof.codebook_symmetry}
If $\mcZ$ is such that the fixed $\mpmX \in \LCMM(\bfr,\mcZ)$, then the for all $(i,j) \in \mcZ$, $X_{ij} = 0$. Since there are $n$ ones in $\mpmX$, $\mcZ$ has to be a subset of the remaining $mn - n$ entries. In addition, since $|\mcZ| = \Kz$, there are a total of $\binom{mn-n}{\Kz}$ number of possible choices. 

For the second result, if $\mcE$ is such that the fixed $\mpmX \in \LCMM(\bfr,\mcE)  $, then the for all $(i,j),(k,l) \in \mcE$, $X_{ij} = X_{kl}$. This means that either $X_{ij} = X_{kl} = 0$ or $X_{ij} = X_{kl} = 1$. The set of all possible fixed-at-equality constraints is then given by $\binom{mn-n}{2} + \binom{n}{2}$. Therefore, this set has $\binom{\binom{mn-n}{2} + \binom{n}{2}}{\Ke}$ size-$\Ke$ subsets.

\subsubsection{Proof of Proposition~\ref{multipermutation.proposition.codesize}}
\label{multipermutation.appendix.proof.codesize}
\begin{align*}
\EX[\size(\LCMM(\bfr,\mcZ))]\! &= \! \sum_{\mcZ \in \Sz(\Kz)} \!\!\!\!  \Pr(\mcZ) 
A(\LCMM(\bfr,\mcZ))\\
&=\! \sum_{\mcZ \in \Sz(\Kz)} \!\!\!\!  \Pr(\mcZ)  \!\!
\sum_{\mpmX \in \mpmset(\bfr)} \!\!\!\! \mathbb{I}(\mpmX \in \LCMM(\bfr,\mcZ))\\
&=\! \frac{1}{|\Sz(\Kz)|} \! \sum_{\mpmX \in \mpmset(\bfr)} \sum_{\mcZ \in \Sz(\Kz)} \!\!\!\! \mathbb{I}(\mpmX \in \LCMM(\bfr,\mcZ)), 
\end{align*}
where $\mathbb{I}(\cdot)$ is the indicator function. By Lemma~\ref{multipermutation.lemma.codebook_symmetry},
$$\sum_{\mcZ \in \Sz(\Kz)} \mathbb{I}(\mpmX \in \LCMM(\bfr,\mcZ)) = \binom{nm -n}{\Kz}.
$$
Then
\begin{align*}
\EX[\size(\LCMM(\bfr,\mcZ))] 
&= \frac{1}{|\Sz(\Kz)|} \sum_{\mpmX \in \mpmset(\bfr)} \binom{nm -n}{\Kz}\\
&= \frac{\binom{nm -n}{\Kz}|\mpmset(\bfr)|}{|\Sz(\Kz)|}.
\end{align*}
In the same vein, we can prove that 
$$
\EX[\size(\LCMM(\bfr,\mcE))] = \frac{\binom{\binom{nm -n}{2} + \binom{n}{2}}{\Ke} |\mpmset(\bfr)|}{|\Se(\Ke)|}.
$$

\subsubsection{Proof of Proposition~\ref{multipermutation.proposition.average_ball}}
\label{multipermutation.appendix.proof.average_ball}
\begin{align*}
&\EX[\ballsize_d(\LCMM(\bfr,\mcZ))] \\
\! &= \! \sum_{\mcZ \in \Sz(\Kz)} \!\!\!\!  \Pr(\mcZ) 
\ballsize_d(\LCMM(\bfr,\mcZ))\\
&=\! \sum_{\mcZ \in \Sz(\Kz)} \!\!\!\!  \Pr(\mcZ)  \!\!
\sum_{\mpmX \in \mpmset(\bfr)} \!\!\!\! \mathbb{I}(\mpmX \in \LCMM(\bfr,\mcZ) \text{ and } \chebyshev(\bft\mpmX, \mpy) \leq d)\\
&=\! \frac{1}{|\Sz(\Kz)|} \! \sum_{\substack{\mpmX \in \mpmset(\bfr) \\\chebyshev(\bft\mpmX, \mpy) \leq d}} \sum_{\mcZ \in \Sz(\Kz)} \!\!\!\! \mathbb{I}(\mpmX \in \LCMM(\bfr,\mcZ)), \\
&=\! \frac{\binom{nm -n}{\Kz}}{|\Sz(\Kz)|} \chebyshevball(r,n,d),
\end{align*}
where $\chebyshevball(r,n,d)$ is the number of elements being $d$-close to a vector in the Chebyshev metric. By Lemma~1-3 in~\cite{shieh2010decoding}, $\chebyshevball(r,n,d)$ can be bounded by 
$$
\frac{(2dr+r)^n n!}{2^{2dr}n^n (r!)^m}\leq  \chebyshevball(r,n,d)
\leq \frac{[(2dr+r)!]^{\frac{n}{2dr+r}}}{(r!)^m}.
$$
Therefore, we obtain~\eqref{multipermutation.eq.ballsize_fixedatzero}. Note that~\eqref{multipermutation.eq.ballsize_fixedatequality} can be obtained in a similar way and thus we omit the details.

\subsection{Numerical results on random coding ensemble}
\label{multipermutation.appendix.numerical_randomcoding}
In this appendix, we compare ST codes with the results obtained in Section~\ref{multipermutation.subsection.randomcoding}. Recall that ST codes are codes with only fixed-at-zero constraints, and we denote the parameters by $d_{ST}$, $m_{ST}$, and $r_{ST}$ (cf. Definition~\ref{multipermutation.def.chebyshev_code} for the role of each of these parameters). For each set of parameters, the cardinality $A_{ST}$ and distance $d_{ST}$ of the corresponding ST code is fixed. Furthermore, each set of parameters corresponds to a unique set of entries that are fixed to zero. We denote this set of entries by $\mcZ_{ST}$. In addition, we let $\Kz_{ST} = |\mcZ_{ST}|$.

We conduct two sets of numerical experiments. In the first set of experiments, we study the scaling of code sizes with the number of fixed-at-zero constraints. First, we compare ST codes with random codes by letting both have the same number of fixed-at-zero constraints. Next, we conduct the reverse experiment by letting both have the same code cardinality. In other words, for each triple $(d_{ST}, m_{ST},r_{ST})$, we let $\Kz_{R}$ be the largest number of fixed-at-zero constraints such that $\EX[\size(\LCMM(\bfr,\mcZ))] \leq A_{ST}$. For both experiments, we scale $d_{ST}$ while fixing $r_{ST}$ and the ratio $m_{ST}/d_{ST}$. Therefore, $m_{ST}$ scales with $d_{ST}$. Since neither experiment compares the distance properties, we conduct a second set of experiments. We first obtain $\Kz_{R}$ in the same way as in the first set of experiments, and then use Proposition~\ref{multipermutation.proposition.average_ball} to bound the average cardinality of radius-$d_r$ balls. In other words, we let the random coding ensemble have the same code cardinality as the ST code, and compare the minimum distances of the codes. We first study the scaling of ball size as a function of $d_{ST}$. Finally, we fix a set of code parameters and obtain the spectrum of ball sizes with respect to the radius $d_r$.

In Figure~\ref{multipermutation.fig.codesizecomparison_scale_d} and~\ref{multipermutation.fig.kappacomparison}, we plot results for the first set of experiments.
In Figure~\ref{multipermutation.fig.codesizecomparison_scale_d}, we define $C_{ST}(d_{ST}) := \log(A_{ST})/d_{ST}$ and plot $C_{ST}$ as a function of $d_{ST}$\footnote{We use the natural logarithm here.}. Furthermore, since each $d_{ST}$ corresponds to a $\Kz_{ST}$, we can generate the random coding ensemble using $\Kz_{ST}$ number of fixed-at-zero constraints. Then, we let $C_{R}(d_{ST}) := \log (\EX[\size(\LCMM(\bfr,\mcZ))] )/d_{ST}$ and plot $C_{R}$ as a function of $d_{ST}$. It is easy to show that $C_{ST} = 18.9405$ for $r = 3$ and $m = 5d_{ST}$, a result that is verified empirically by Figure~\ref{multipermutation.fig.codesizecomparison_scale_d}. However, we observe from Figure~\ref{multipermutation.fig.codesizecomparison_scale_d} that $C_{R}(d_{ST})$ decreases as $d_{ST}$ increases. Nevertheless, we can show that $\lim_{d_{ST} \rightarrow \infty} C_{R}(d_{ST}) \geq 13.31$, which means that $\EX[\size(\LCMM(\bfr,\mcZ))]$ scales exponentially with $d_{ST}$ asymptotically.
We note that even at small $d_{ST}$, ST codes are much larger than the ensemble average. For example, when $d_{ST} = 5$, number of codewords in the ST code is $10^6$ times larger than that of the ensemble average. Interestingly, in the reverse experiment, we observe from Figure~\ref{multipermutation.fig.kappacomparison} that the number of fixed-at-zero constraints for ST codes does not need to be significantly larger than the ensemble average in order to make both have the same number of codewords. However, changing the number of fixed at zero constraints does have a significant impact on the distance property of the code, as we demonstrate next.

\begin{figure}[!htbp]
\psfrag{&STcodesize}{\scriptsize{ST code}}
\psfrag{&ensembleaverage}{\scriptsize{Ensemble average}}
\psfrag{&dst}{\scriptsize{$d_{ST}$}}
\psfrag{&size}{\hspace{-1cm}\scriptsize{code size scaling factors}}
	\begin{center}
    \includegraphics[width = 21pc]{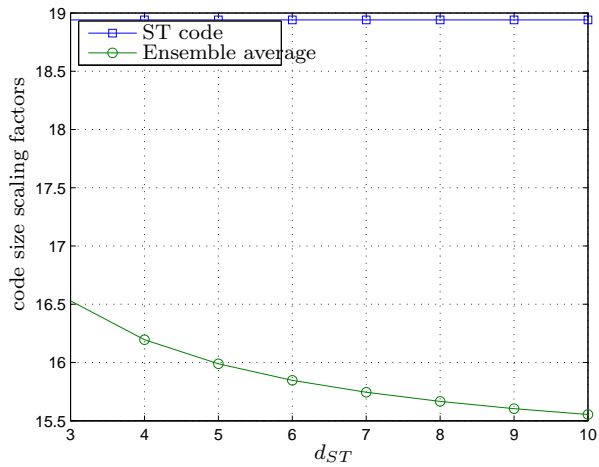}
    \end{center}
    \caption{$C_{ST}:= \log(A_{ST})/d_{ST}$ and $C_{R}:= \log (\EX[\size(\LCMM(\bfr,\mcZ))] )/d_{ST}$ plotted as a function of $d_{ST}$. Both have the same number of fixed-at-zero constraints. $r = 3$ and $m = 5d_{ST}$.}
\label{multipermutation.fig.codesizecomparison_scale_d}
\end{figure}

\begin{figure}[!htbp]
\psfrag{&matrixsizepadpad}{\scriptsize{Matrix size ($mn$)}}
\psfrag{&kappast}{\scriptsize{$\Kz_{ST}$}}
\psfrag{&kapparan}{\scriptsize{$\Kz_{R}$}}
\psfrag{&dst}{\scriptsize{$d_{ST}$}}
\psfrag{&numberofzeros}{\hspace{0.5cm}\scriptsize{$|\mcZ|$}}
	\begin{center}
    \includegraphics[width = 21pc]{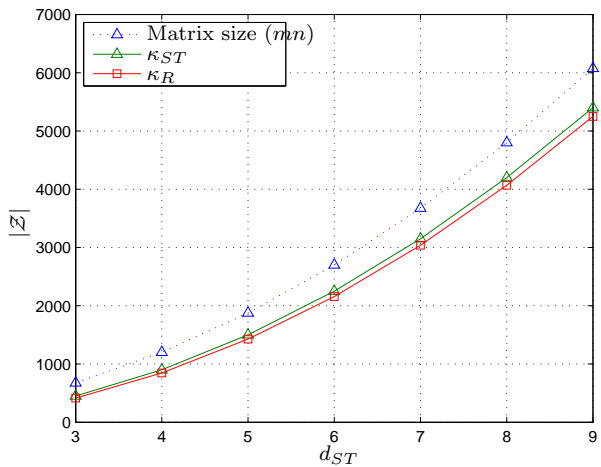}
    \end{center}
    \caption{$\Kz_{ST}$ and $\Kz_{R}$ plotted as a function of $d_{ST}$ for random fixed-at-zero codes with parameters $r = 3$ and $m = 5d_{ST}$.}
\label{multipermutation.fig.kappacomparison}
\end{figure}

In Figure~\ref{multipermutation.fig.ballsize_scale_d} and~\ref{multipermutation.fig.ballsize_scale_radius}, we plot the results of the second set of experiment. In Figure~\ref{multipermutation.fig.ballsize_scale_d}, we scale $d_{ST}$ while keeping  $r$ fixed at $3$ and $m$ at $5d_{ST}$. We make the following remarks. First, recall that $\Kz_{ST}$ is larger than $\Kz_{R}$ by a relatively small number. However, we observe in Figure~\ref{multipermutation.fig.ballsize_scale_d} that the expected number of codewords within radius $d_{ST}$ of codes with fixed-at-zero constraints, $\EX[\ballsize_{d_{ST}}(\LCMM(\bfr,\mcZ))]$, is quite large. Note that for ST codes, there is \textbf{no} codeword that is $d_{ST}$-close to any other codeword. This indicates that the distance property of ST codes is much better than the ensemble average. Second, we observe that the upper and lower bounds still have room for improvements. In particular, there is an exponentially increasing gap between the two. In Figure~\ref{multipermutation.fig.ballsize_scale_radius}, we pick the ST code defined by parameters $d_{ST} = 5$,  $r = 3$, $m = 30$, and $n = 90$. Using this set of parameters, we find that random code design with $\Kz_{R} = 2158$ zeros achieves the same average codebook size as the ST code, i.e., around $2.26\cdot 10^{49}$. We observe in Figure~\ref{multipermutation.fig.ballsize_scale_radius} that the upper bound on the average ball size is less than $1$ for $d_r \leq 3.5$. This means that on average, there is less than $1$ codeword within the radius-$3.5$ ball of any codeword. Although this result does not directly translate to the minimum Chebyshev distance of a code, we can still infer that the minimum distance is around $3$, which is less than the minimum distance of the ST code, which is $5$. 
\begin{figure}[!htbp]
\psfrag{&upperboundpad}{\scriptsize{Upper bound}}
\psfrag{&lowerbound}{\scriptsize{Lower bound}}
\psfrag{&dst}{\scriptsize{$d_{ST}$}}
\psfrag{&averageballsize}{\hspace{-0.5cm}\scriptsize{$\EX[\ballsize_{d_{ST}}(\LCMM(\bfr,\mcZ))]$}}
	\begin{center}
    \includegraphics[width = 21pc]{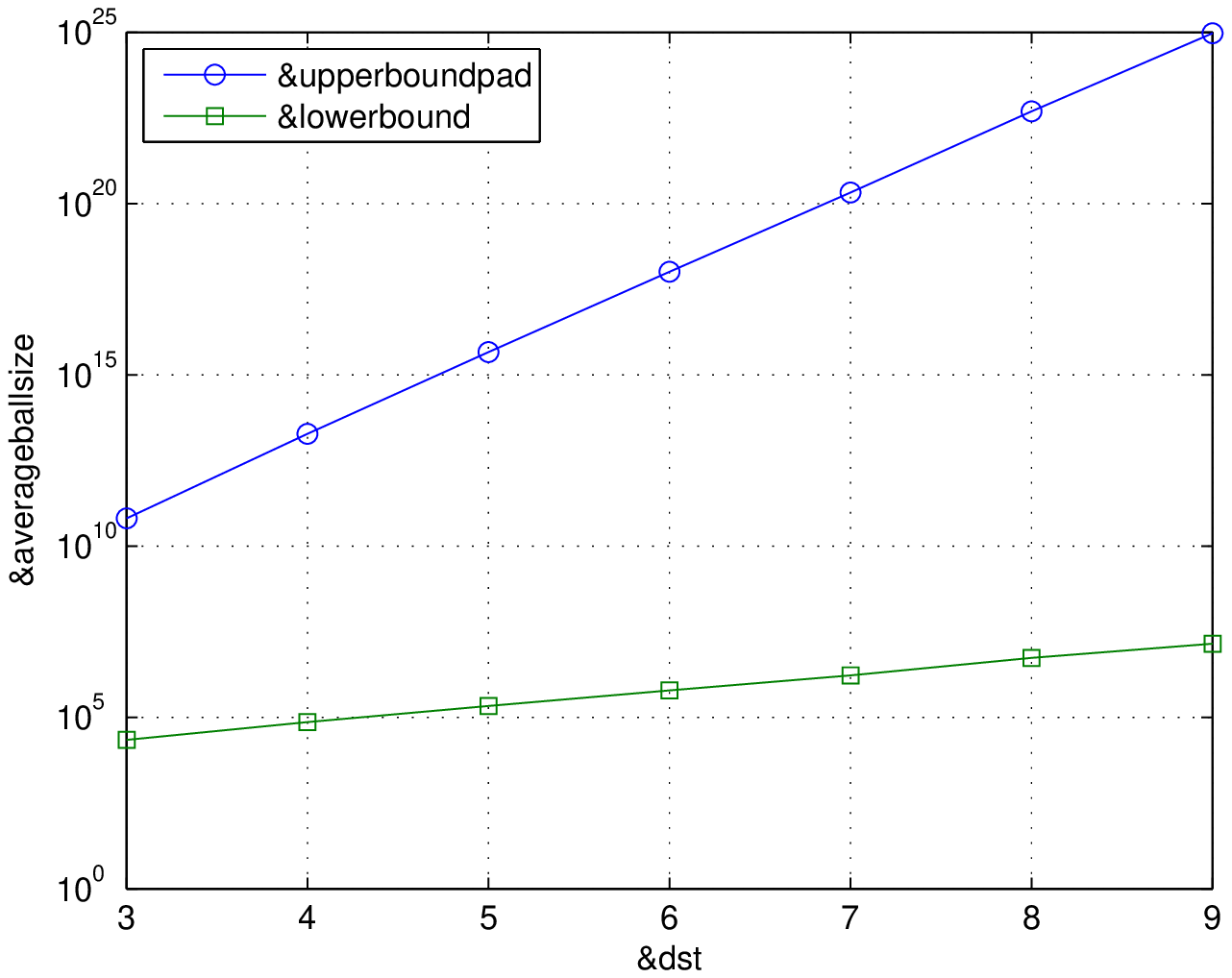}
    \end{center}
    \caption{Upper and lower bound of $\EX[\ballsize_{d_{ST}}(\LCMM(\bfr,\mcZ))]$ plotted as a function of $d_{ST}$ for random fixed-at-zero codes with parameters $r = 3$ and $m = 5d_{ST}$.}
\label{multipermutation.fig.ballsize_scale_d}
\end{figure}

\begin{figure}[!htbp]
\psfrag{&upperboundpad}{\scriptsize{Upper bound}}
\psfrag{&lowerbound}{\scriptsize{Lower bound}}
\psfrag{&ballradius}{\scriptsize{Ball radius}}
\psfrag{&averageballsize}{\hspace{-0.5cm}\scriptsize{$\EX[\ballsize_{d_{ST}}(\LCMM(\bfr,\mcZ))]$}}
	\begin{center}
    \includegraphics[width = 21pc]{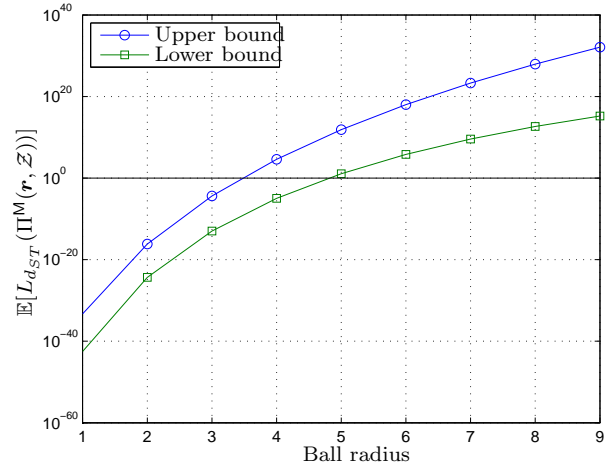}
    \end{center}
    \caption{Upper and lower bound of $\EX[\ballsize_{d_{ST}}(\LCMM(\bfr,\mcZ))]$ plotted as a function of $d$ for random fixed-at-zero codes with parameters $d_{ST} = 5$,  $r = 3$, $m = 30$, and $n = 90$. $\Kz_{R} = 2158$ to match the code size of the ST code defined by the same parameters.}
\label{multipermutation.fig.ballsize_scale_radius}
\end{figure}
\section{Projection onto the $\ell_1$ ball with box constraints}
\label{multipermutation.appendix.linear_time_projection}
In this appendix, we show a linear time projection algorithm onto $\oneslice_{n}^{r}$. There are two key ideas in this algorithm. First, using Karush-Kuhn-Tucker (KKT) conditions, the projection problem can be transformed to a waterfilling type problem, wherein one needs to perform a binary search over $2n$ possible points. Second, although one can first sort these points and then perform a binary search, the sorting operation is $O(n\log n)$ and is expensive. Instead, the binary search can be done based on linear time median finding algorithms (e.g.,~\cite[Sec.~8.5]{press2007numerical}). These two key ideas are used in both~\cite{duchi2008efficient} and~\cite{gupta2010l1}.

Algorithm~\ref{multipermutation.algorithm.linear_time_projection} in this section is modified from~\cite[Algorithm~1]{gupta2010l1} and hence is not new. We note that Barman \emph{et al.} also worked on the problem and present their algorithm in~\cite{barman2011decomposition}. However, their algorithm is based on sorting instead of median finding. Due to these previous works, we only briefly describe the derivations. Our goals in this section are to first correct some errors in~\cite[Algorithm~1]{gupta2010l1}\footnote{The geometry described in~\cite{gupta2010l1} corresponds to the inequality constraint $\sum_{i = 1}^{n} x_i \leq r$. However,~\cite[Algorithm~1]{gupta2010l1} actually projects onto the geometry with the equality constraint. In addition, there are several minor errors with~\cite[Algorithm~1]{gupta2010l1}.}, and second to compare with the sorting based algorithm proposed in~\cite{barman2011decomposition}. We refer readers to~\cite{barman2011decomposition} for a nice waterfilling interpretation of the algorithm. Also note that switching from sorting to median finding requires tracking partial sums in each iteration, which is described in~\cite{gupta2010l1}. 

Projection onto $\oneslice_{n}^{r}$ is equivalent to the following optimization problem.
\begin{equation}
\label{multipermutation.eq.projection_problem_eq}
\begin{split}
\opmin \quad & \frac{1}{2}\|\bfv - \bfx\|_2^2  \\
\st\quad  & 0\leq x_i \leq 1\text{ for all }i \text{ and } \sum x_i = r.
\end{split}
\end{equation}
We introduce multipliers $\theta$, $\eta_i$ and $\nu_i$, and write the Lagrangian of~\eqref{multipermutation.eq.projection_problem_eq} as
\begin{equation*}
\begin{split}
\Lag(\bfx, \mu, \eta, \nu) =& \sum_i \frac{1}{2}(v_i - x_i)^2 + \theta \left(r - \sum_i x_i\right) \\ &- \sum_i\eta_i(1 - x_i) - \sum_i\nu_ix_i.
\end{split}
\end{equation*}
Denote by $\bfx^*$ an optimal solution and denote by $\theta^*$, $\eta_i^*$, and $\nu_i^*$ the corresponding multipliers. The KKT conditions imply $\nabla\Lag(\bfx^*) = \bfzero$, $v_i - x_i^* = \theta^* + \eta^*_i - \nu^*_i$, $\eta_i^* (1- x_i^*) =0 $, and $\nu_i^* x_i^* = 0$. This means that for all $0 < x_i^* < 1$, $v_i - x_i^* = \theta^*$. In other words, for a fixed $\theta$, the indices $\{1,\dots,n\}$ are divided into three sets: 
\begin{itemize}
\item An \emph{active set} $\activeset$ such that for $i\in\activeset$, $0 < v_i - \theta < 1$ and $x_i = v_i - \theta$.
\item A \emph{clipped set} $\clippedset$ such that for $i\in\clippedset$,  $v_i - \theta > 1$ and $x_i = 1$.
\item A \emph{zero set} $\zeroset$ such that for $i\in\zeroset$, $v_i - \theta < 0$ and $x_i = 0$.
\end{itemize}
When $\theta$ varies in a range that does not change the three sets, $\sum_i x_i$ becomes a linear function with respect to $\theta$. On the other hand, there are some values of $\theta$ at which these sets change, which we term \emph{break points}. The set of break points is easy to identify, it can be defined by $\mcB:= \bigcup_{i = 1,\dots,n} (\{v_i, v_i - 1\}) $. In Algorithm~\ref{multipermutation.algorithm.linear_time_projection}, we perform a binary search over all break points, each corresponds to a triplet of active, clipped, and zero sets. In addition, the pivot for each binary search can be determined by a median finding algorithm. In Figure~\ref{multipermutation.fig.lineartimeprojection}, we compare Algorithm~\ref{multipermutation.algorithm.linear_time_projection} with the sorting based algorithm proposed in~\cite{barman2011decomposition}. It is easy to observe that the linear time algorithm is significantly faster.
\begin{figure}[!htbp]
\psfrag{&sortingbasedprojection}{\scriptsize{Sorting based proj.}}
\psfrag{&lineartimeprojection}{\scriptsize{Linear time proj.}}
\psfrag{&projectionsize}{\hspace{-0.5cm}\scriptsize{Projection dimension ($n$)}}
\psfrag{&timemilliseconds}{\hspace{-0.5cm}\scriptsize{time (milliseconds)}}
	\begin{center}
    \includegraphics[width = 21pc]{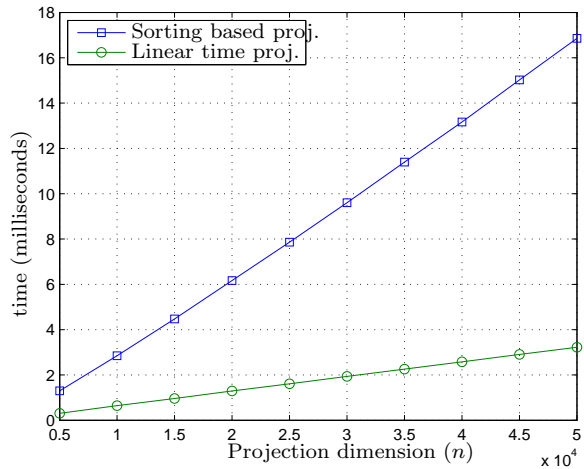}
    \end{center}
    \caption{Execution time plotted as a function of projection dimension. Data collected on an Intel(R) 2.5GHz CPU.}
\label{multipermutation.fig.lineartimeprojection}
\end{figure}

\subsubsection*{Remarks on Algorithm~\ref{multipermutation.algorithm.linear_time_projection}}
\begin{itemize}
\item Step~\ref{alg.step.r} calculates the $\ell_1$ norm of the current projection vector determined by $\theta_p$ (cf.~\cite[Eq.~(11)]{gupta2010l1}).
\item Taking the median of the first few entries of $\mcB$ can be good enough in practice (e.g, first $50$). This avoids taking the median of a huge set of numbers and which accelerate the algorithm for large values of $n$.\footnote{This approach uses a heuristic approximation of the true median. As a result, the total number of iterations may increase, but the overall execution time may decrease.}
\end{itemize}
\begin{algorithm}[!htbp]
\caption{Project vector $\bfv$ onto $\oneslice_{n}^{r}$}
\label{multipermutation.algorithm.linear_time_projection}
\begin{algorithmic}[1]
\STATE Construct the set $\mcB:= \bigcup_{i = 1,\dots,n} (\{v_i, v_i - 1\})$.
\STATE Construct the set $\mcV^{\mathrm{uncertain}}:= \{v_1, \dots, v_n\}$.
\STATE $n_{\mathrm{clip}} \leftarrow 0$, $n_{\mathrm{zero}} \leftarrow 0$; $s_{\mathrm{clip}} \leftarrow 0$, $s_{\mathrm{zero}} \leftarrow 0$, $s_{\mathrm{all}} \leftarrow \sum_{i = 1}^n v_i$.
\WHILE{$|\mcB| > 2$} 
\STATE $\theta_p \leftarrow $\texttt{median}$(\mcB)$.
\STATE Use $\theta_p$ to partition $\mcV^{\mathrm{uncertain}}$ into active set $\activeset'$, clipped set $\clippedset'$, and zero set $\zeroset'$.
\STATE $n_{\mathrm{clip}}' \leftarrow |\clippedset'|$, $s_{\mathrm{clip}}' \leftarrow \texttt{sum}(\mcV^{\mathrm{uncertain}}_{\clippedset'}) - n_{\mathrm{clip}}'$. $n_{\mathrm{zero}}' \leftarrow |\zeroset'|$, $s_{\mathrm{zero}}' \leftarrow \texttt{sum}(\mcV^{\mathrm{uncertain}}_{\zeroset'})$.
\STATE\label{alg.step.r} Evaluate \vspace{-3mm}
\begin{align*}
r_{\mathrm{current}} = &s_{\mathrm{all}} - s_{\mathrm{zero}} - s_{\mathrm{clip}} - \theta_p (d - n_{\mathrm{clip}} - n_{\mathrm{zero}}) \\&- s_{\mathrm{clip}}' - s_{\mathrm{zero}}' + \theta_p ( n_{\mathrm{clip}}' + n_{\mathrm{zero}}')
\end{align*}\vspace{-5mm}
\IF {$r_{\mathrm{current}} > r$}
\item[] \{Increase $\theta_p$. Fix the current zero set.\}
    \STATE  $s_{\mathrm{zero}} \leftarrow s_{\mathrm{zero}} + s_{\mathrm{zero}}'$,
			$n_{\mathrm{zero}} \leftarrow n_{\mathrm{zero}} + n_{\mathrm{zero}}'$.
	\STATE  Remove $\zeroset'$ entries from $\mcV^{\mathrm{uncertain}}$.
	\STATE	Update $\mcB$ by deleting elements less than $\theta_p$.
\ELSIF {$r_{\mathrm{current}} < r$}
\item[] \{Decrease $\theta_p$. Fix the current clipped set.\}
	\STATE  $s_{\mathrm{clip}} \leftarrow s_{\mathrm{clip}} + s_{\mathrm{clip}}'$,
			$n_{\mathrm{clip}} \leftarrow n_{\mathrm{clip}} + n_{\mathrm{clip}}'$.
	\STATE  Remove $\clippedset'$ entries from $\mcV^{\mathrm{uncertain}}$.
	\STATE	Update $\mcB$ by deleting elements greater than $\theta_p$.
\ELSE 
    \STATE $\theta^* \leftarrow \theta_p$. \COMMENT{Success.}
    \STATE Determine the projection by applying $\theta^*$ to the KKT conditions. \textbf{Return.}
\ENDIF
\ENDWHILE
\COMMENT{$\mcB$ has two elements.}
\STATE Using $\max \mcB$, evaluate $\activeset$, $\clippedset$, and $\zeroset$ for $\bfv$.
\STATE $\theta^* \leftarrow \frac{|\clippedset| + \texttt{sum}(\bfv_{\activeset}) - r}{|\activeset|}$. Determine the projection by applying $\theta^*$ to the KKT conditions. \textbf{Return.}
\end{algorithmic}
\end{algorithm}
\section{Estimating the initial vector for rank modulation}
\label{multipermutation.appendix.initial_vector}
A key outcome of this paper is that soft decoding can outperform hard decoding and thus is promising in practice. However, it may be infeasible to obtain information required to set up the soft decoding problem as stated. For instance, this may be the situation in flash memory, where the initial vector is determined at the cell programming stage and is not fixed. Consequently, the decoder does not know the exact initial vector and hence cannot directly apply the soft decoding techniques developed in this paper. In this appendix, we initiate the study of these issues by presenting a model for estimating the initial vector given certain constraints on its uncertainty. 

Let $\code$ denote a codebook of multipermutation matrices. In rank modulation, the encoder encodes a message into a multipermutation ranking by injecting charge to memory cells. In this case, the initial vector $\bft$ is determined after the memory cells are programmed to the desired ranking. Let $\mpx$ be the actual charge levels stored using rank modulation, then $\mpx = \bft \mpmX$ where $\mpmX \in \code$. Due to the physical resolution limitations of flash memories, we assume that $|t_i - t_j| \geq \Delta$ for all $i,j \in \{1,\dots,m\}$. The value $\Delta$ is assumed known by the decoder. In addition, without loss of generality, we may assume that $t_1 < t_2 < \dots < t_m$. In summary, only $\code$ (which includes the values of $\bfr$, $m$, and $n$), and $\Delta$ are known to the decoder.

For simplicity, we assume that the noise is additive and the decoder observes channel output $\bfy = \mpx + \bfn$. In this scenario, the decoder needs to solve the following ML decoding problem
\begin{align*}
\max & \quad \Pr[\bfy|\bft \mpmX]\\
\st & \quad |t_i - t_j| \geq \Delta \text{ for all $i$ and $j$},\\
&\quad \bft \text{ is sorted in ascending order},\\
&\quad \mpmX \in \code.
\end{align*}
This problem involves multiplying two variables ($\bft$ and $\mpmX$) and thus cannot be written as a linear program. We explore two options to address this problem.

\subsection{Restricting the initial vector}
One natural idea for this problem is to enforce more constraints on $\bft$. For example, we can require that $t_{i+1} - t_{i} = \Delta$, where $\Delta$ is a constant known by both the encoder and the decoder. As a result, $\bft$ can be represented as $\bft = \bft^N + \eta$, where $\bft^N = (\Delta, 2\Delta, \dots, (m-1)\Delta, m\Delta)$ is a normalized vector and $\eta$ is a constant chosen by the encoder but not known to the decoder. Consequently, $\bfy = \bft\bfX + \bfn = \bft^N \bfX + \eta + \bfn$. Therefore, the decoder can perform the following two-step decoding.
\begin{enumerate}
\item Estimate $\eta$ by $\hat{\eta} = \frac{1}{n} (\sum_{j = 1}^n y_j  - \sum_{i = 1}^m r_it^N_i)$.
\item Soft decoding (e.g. LP decoding) using $\hat{\eta}$.
\end{enumerate}
Although this method is easy, one needs to develop a matching write process for such initial vectors. In particular, over-injections and ranking modifications need to be taken care of. In these two scenarios, one needs to increase all other cells so that they satisfy the condition $t_{i+1} - t_{i} = \Delta$. On the other hand, this approach can reduce the number of rewrites of memories due to the restrictive choices of $\bft$.

\subsection{Turbo-equalization like decoding for initial vectors on grids}
We slightly relax the previous, restrictive, conditions and consider the case where all initial vectors take on values that are multiples of $\Delta$. In other words, $t_i = k_i \Delta$, where $k_i \in \mathbb{Z}^+$ and $k_{i+1} \geq k_{i} + 1$. Furthermore, we assume the following largest cell condition: $t_m - t_{m-1} = \Delta$. The reason behind this assumption is that the ranking of the cells stay the same as long as $t_m \geq t_{m-1}  + \Delta$. Thus, increasing $t_m$ during cell programming can reduce the number of rewrites before a block erasure. By leveraging these conditions, we propose the following iterative \emph{turbo-equalization} decoding algorithm.
\begin{enumerate}
\item[(1)] Decode using quantized ranking (by LP or bounded distance decoding), round all fractional solutions, and obtain $\hat{\mpmX}$.
\item[(2)] Let $\bft^*$ be the solution of the following problem.
\begin{equation}
\label{multipermutation.eq.estimate_initial_vector}
\begin{aligned}
\max & \quad \Pr[\bfy|\bft \hat{\mpmX}]\\
\st & \quad t_{i+1} - t_{i} \geq \Delta \text{ for all }i = 1,\dots,m-2,\\
&\quad t_m - t_{m-1} = \Delta, \text{ and }t_1\geq 0.
\end{aligned}
\end{equation}
\item[(3)] Round $\bft^*$ to the nearest multiples of $\Delta$. Denote it by $\hat{\bft}$.
\item[(4)] Use $\hat{\bft}$ to perform soft decoding and update the estimate $\hat{\mpmX}$.
\item[(5)] Repeat Step 2 to 4, each time with the latest estimates.
\end{enumerate}

For the AWGN channel, Step (2) can be formulated as a quadratic program and is solvable using off-the-shelf solvers. In Figure~\ref{multipermutation.fig.TurboSimulation}, we simulate the same code that we use in Figure~\ref{multipermutation.fig.STCodes_1-161-161-16}. The true initial vector is $(1,2,\dots,16)$. In addition, $\Delta = 1$. In this experiment, we only use LP decoding of Chebyshev distance. The difference among the curves is the input to each decoder. For turbo-equalization decoding, we use hard LP decoding of Chebyshev distance based on quantized ranking, followed by one or two turbo iterations. In each turbo iteration, we first estimate the initial vector using~\eqref{multipermutation.eq.estimate_initial_vector}, and then decode using soft LP decoding of Chebyshev distance based on the estimated initial vector. Note that the turbo-equalization decoder is not provided with the true initial vector. The other two curves are baseline curves from Figure~\ref{multipermutation.fig.STCodes_1-161-161-16}.   
\begin{figure}
\psfrag{SNR}{\scriptsize{SNR (dB)}}
\psfrag{WER}{\hspace{-0.5cm}\scriptsize{word-error-rate}}
\psfrag{&LPDecChebyshevHard------------}{\scriptsize{LP Chebyshev, hard}}
\psfrag{&LPDecChebyshevSoft}{\scriptsize{LP Chebyshev, soft}}
\psfrag{&LPDecChebyshevTurbo1}{\scriptsize{Turbo, LP Chebyshev: $1$ iteration}}
\psfrag{&LPDecChebyshevTurbo2}{\scriptsize{Turbo, LP Chebyshev: $2$ iterations}}
	\begin{center}
    \includegraphics[width = 21pc]{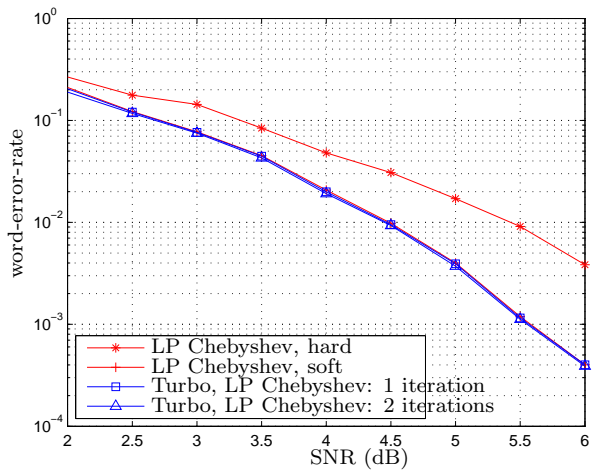}
    \end{center}
    \caption{Word-error-rate (WER) of turbo-equalization decoding plotted as a function of signal-to-noise ratio (SNR) for the ST code defined by parameter $r = 3$, $d = 4$, and $m = 16$. The codeword transmitted is $(1,\dots,16,1,\dots,16,1,\dots,16)$. }
    \label{multipermutation.fig.TurboSimulation}
\end{figure}

We observe that the turbo-equalization decoding technique performs as well as when the true initial vector is provided. In addition, one turbo iteration is sufficient for the case simulated in Figure~\ref{multipermutation.fig.TurboSimulation}. Further, we observe that using erroneous initial vectors in soft LP decoding of Chebyshev distance almost always yields incorrect decoding results (data not shown). On the other hand, we also observe that the initial vector estimation is not always accurate.

We believe that the algorithm proposed is promising, but there are many open questions. Two important ones are first how to analyze the error performance of turbo-equalization decoding and second what classes of codes that benefit from this scheme.

\end{document}